\documentclass[11pt,a4paper]{article}

\usepackage[utf8]{inputenc}
\usepackage[T1]{fontenc}
\usepackage{lmodern}

\usepackage{amsmath,amssymb,amsthm}
\usepackage{bm}              % Bold math symbols (\bm)
\usepackage{mathtools}        % Extensions to amsmath

\usepackage{booktabs}
\usepackage{array}
\usepackage{tabularx}

\usepackage[margin=2.5cm]{geometry}
\usepackage{setspace}
\usepackage{parskip}          % Paragraph spacing instead of indentation

\usepackage[authoryear,round]{natbib}

\usepackage{tikz}
\usepackage{pgfplots}
\pgfplotsset{compat=1.18}
\usetikzlibrary{arrows.meta, decorations.markings, patterns, calc}

\usepackage[colorlinks=true,linkcolor=blue,citecolor=blue,urlcolor=blue]{hyperref}
\usepackage{cleveref}

\newtheorem{theorem}{Theorem}
\newtheorem{proposition}{Proposition}
\newtheorem*{proposition*}{Proposition}  % Unnumbered variant
\newtheorem{corollary}{Corollary}
\theoremstyle{remark}
\newtheorem*{remark}{Remark}

\newcommand{\w}{\mathbf{w}}
\newcommand{\R}{\mathbf{R}}
\newcommand{\one}{\mathbf{1}}
\newcommand{\V}{\mathbf{V}}
\newcommand{\Z}{\mathbf{Z}}
\newcommand{\Sig}{\boldsymbol{\Sigma}}
\newcommand{\muvec}{\boldsymbol{\mu}}
\newcommand{\emu}{\mathbf{e}_\mu}
\newcommand{\CV}{\mathrm{CV}}
\newcommand{\SD}{\mathrm{SD}}
\newcommand{\SR}{\mathrm{SR}}
\newcommand{\Cov}{\mathrm{Cov}}
\newcommand{\Var}{\mathrm{Var}}
\newcommand{\tw}{\tau_w}
\newcommand{\rf}{r_f}
\DeclareMathOperator*{\argmax}{arg\,max}

\title{Asset Returns, Portfolio Choice, and Proportional Wealth Taxation}
\author{Anders G Fr{\o}seth\thanks{Independent Researcher.
  E-mail: \href{mailto:indrefjorden@pm.me}{indrefjorden@pm.me}.}}
\date{\today}

\begin{document}
\maketitle

\begin{abstract}
We analyse the effect of a proportional wealth tax on asset returns, portfolio
choice, and asset pricing in a partial equilibrium setting. The tax is levied
annually on the market value of all holdings at a uniform rate. We show that
such a tax is economically equivalent to the government acquiring a
proportional stake in the investor's portfolio each period---a form of risk
sharing in which expected wealth and risk are reduced by the same factor,
while the return per share is unaffected. This multiplicative separability
between the tax factor and the return realisation drives four main results.
First, the coefficient of variation of wealth is invariant to the tax rate,
since the tax reduces expected wealth and risk by the same proportion. Second,
the optimal portfolio weights---and in particular the tangency
portfolio---are independent of the tax rate. Third, the wealth tax is
orthogonal to portfolio choice: in discrete time it induces a homothetic
contraction of the opportunity set in the mean--standard deviation plane that
preserves the Sharpe ratio of every portfolio. Fourth, both taxed and untaxed
investors are willing to pay the same price per share for any asset. The
results are derived first under geometric Brownian motion and then generalised
to any return distribution in the location-scale family. A complementary
Modigliani-Miller analysis, treating the tax claim as a separate security,
confirms pricing neutrality and identifies an inconsistency in the existing
literature regarding the discount rate used for after-tax cash flows.
Imposing the Capital Asset Pricing Model as a special case confirms that
after-tax betas equal pre-tax betas and the security market line contracts
uniformly by $(1-\tw)$; under CRRA preferences, general-equilibrium
returns and prices are unchanged. This resolves an error in
\citet{Fama2021}, who overstates the price effect by adding the wealth tax
to the cost of capital without adjusting the discount rate.
The neutrality results depend on two conditions that are commonly violated
in practice: universal taxation at market value, and frictionless markets.
We formalise three channels through which relaxing these conditions breaks
neutrality---book-value taxation, liquidity frictions, and dividend
extraction---and show that they have opposing effects on asset prices.
\end{abstract}



%\tableofcontents
%\newpage

% =========================================================================
\section{Introduction}\label{sec:intro}

This paper develops a framework for analysing the effect of a proportional
wealth tax on asset returns, portfolio choice, and asset pricing. We compare
two investors who are identical in all respects except that one (Investor~A)
pays an annual proportional wealth tax on the market value of all holdings,
while the other (Investor~B) does not.

The main results are:
\begin{enumerate}
  \item A proportional wealth tax on all assets is economically equivalent to
        the government acquiring a proportional claim on the investor's
        assets, sharing both risk and return pro rata. From the investor's
        cash-flow perspective, this amounts to a periodic partial sale of
        shares to meet the tax obligation.
  \item The tax does not alter the return distribution per unit of asset, the
        risk-reward profile (coefficient of variation), or the optimal
        portfolio weights.
  \item The wealth tax is orthogonal to portfolio choice: it operates purely
        in the return dimension while portfolio optimisation operates in the
        risk dimension.
  \item Both investors should be willing to pay the same price per share for
        any asset.
\end{enumerate}

These results hold under the economic assumptions of proportional taxation on
all assets and partial equilibrium. We develop them in two stages. In
Sections~\ref{sec:single_gbm}--\ref{sec:gbm_summary}, we derive the results
under geometric Brownian motion (GBM), which provides a concrete and tractable
setting with closed-form expressions for the moments of the wealth
distribution. In \Cref{sec:generalisation}, we show that the GBM assumption
can be substantially relaxed: the results hold for any return distribution in
the location-scale family (and for some results, any distribution with finite
second moments). The distributional assumption turns out to play no
substantive economic role---the structural engine behind all four results is
the multiplicative separability of the wealth tax from the return realisation.

\Cref{sec:conditions} discusses the conditions under which the results hold
and clarifies their independence from any specific asset pricing model. A
complementary analysis under the CAPM (\Cref{sec:capm_specialcase}) confirms
that after-tax betas and the security market line are preserved, that
general equilibrium is undisturbed under CRRA preferences, and identifies a
pricing error in \citet{Fama2021}. A Modigliani-Miller perspective
completes the discussion; a survey of the related literature is collected
in \Cref{app:literature}.
\Cref{sec:beyond_neutrality} then asks what happens when the
key conditions fail. Three channels of non-neutrality are formalised:
book-value taxation (where the tax base diverges from market value),
liquidity frictions (where forced selling incurs transaction costs), and
dividend extraction (where the tax forces payouts that displace profitable
investment). These channels have opposing effects on asset prices, and their
net impact depends on the asset class. \Cref{sec:conclusion} summarises the
contributions and outlines directions for further theoretical and empirical
work. \Cref{app:dividends} analyses alternative tax payment mechanisms.

% =========================================================================
\section{Setup}\label{sec:setup}

\subsection{Asset Return Dynamics}\label{sec:setup_dynamics}

There are $K$ risky assets with prices following a multivariate GBM:
\begin{equation}\label{eq:gbm_scalar}
  \frac{dP_i}{P_i} = \mu_i \, dt + \sum_j \sigma_{ij} \, dZ_j,
  \qquad i = 1, \ldots, K.
\end{equation}
In vector notation:
\begin{equation}\label{eq:gbm_vector}
  \frac{d\mathbf{P}}{\mathbf{P}} = \muvec \, dt + \Sig \, d\Z
\end{equation}
where $\muvec$ is the $K \times 1$ vector of expected returns and
$\V = \Sig\Sig^\top$ is the $K \times K$ covariance matrix. A risk-free asset
with continuous return $\rf$ is also available.

The GBM assumption will be relaxed in \Cref{sec:generalisation}, where we show
that the main results depend only on the existence of well-defined first and
second moments. In the multivariate setting, the natural distributional class
is the family of elliptical distributions
\citep{OwenRabinovitch1983}, which nests the multivariate normal (and hence
GBM) as a special case.

\subsection{Tax Structure}\label{sec:setup_tax}

We consider a single tax instrument: a personal wealth tax at rate
$\tw \in (0,1)$, levied annually on the market value of all assets held by the
investor (including the risk-free asset). There is no income tax and no capital
gains tax. Two features of this specification are worth highlighting.
First, the tax is \emph{proportional}: the rate~$\tw$ is the same regardless
of wealth level. Second, the tax is \emph{universal}: all assets are taxed at
the same rate on market value, so no asset bears a higher or lower effective
tax burden than any other. Together, proportionality and universality generate
the multiplicative separability that drives the main results; relaxing either
opens channels of non-neutrality (\Cref{sec:beyond_neutrality}).

For expositional clarity, dividends are assumed to be paid out at year-end
and consumed; the wealth tax is then paid by selling shares.  This is a
simplifying convention, not a necessary condition for neutrality: if
dividends are instead reinvested in assets that remain within the tax base,
the multiplicative separability is preserved and all results carry through
(\Cref{app:dividends}).

\subsection{Investors}\label{sec:setup_investors}

\textbf{Investor~A} pays wealth tax $\tw > 0$ on all asset holdings.
\textbf{Investor~B} pays no wealth tax ($\tw = 0$). The two investors are
otherwise identical in preferences, information, and endowments.

\subsection{Notation}\label{sec:notation}

\Cref{tab:notation} collects the principal symbols used throughout the paper.

\begin{table}[ht]
\centering
\caption{Notation guide.}
\label{tab:notation}
\small
\begin{tabular}{@{}ll@{\qquad}ll@{}}
\toprule
\multicolumn{2}{l}{\textit{Assets and returns}}
  & \multicolumn{2}{l}{\textit{Portfolio}} \\
\midrule
$K$
  & Number of risky assets
  & $\w$
  & Portfolio weight vector \\
$P_i$, $P_t$
  & Asset price
  & $\w^*$
  & Optimal weights \\
$\muvec$, $\mu_i$
  & Expected return (vector / scalar)
  & $\w_T$
  & Tangency portfolio \\
$\Sig$
  & Diffusion (volatility) matrix
  & $\mu_P(\w)$
  & Portfolio expected return \\
$\V = \Sig\Sig^\top$
  & Covariance matrix
  & $\sigma_P(\w)$
  & Portfolio volatility \\
$\rf$
  & Risk-free rate
  & $R_P(\w)$
  & Portfolio rate of return (net) \\
$\R$
  & Return vector $(R_1,\ldots,R_K)^\top$
  & $G_P$
  & Portfolio gross return $1+R_P$ \\
$\Z$
  & Standard Brownian motion /
  & &\\
  & \quad standardised random vector
  & &\\
$G^{(n)}$
  & Cumulative gross return $P_n/P_0$
  & &\\[6pt]
\midrule
\multicolumn{2}{l}{\textit{Wealth, tax, and investors}}
  & \multicolumn{2}{l}{\textit{After-tax quantities}} \\
\midrule
$W_0$, $W_n$
  & Wealth at time $0$, $n$
  & $R_W(\w,\tw)$
  & After-tax return on wealth \\
$W_n^A$, $W_n^B$
  & Wealth of Investor~A, B
  & $\mu_W(\w,\tw)$
  & After-tax expected return \\
$N_0$, $N_n$
  & Number of shares at time $0$, $n$
  & $\sigma_W(\w,\tw)$
  & After-tax volatility \\
$\tw$
  & Wealth tax rate
  & $\rf^A$
  & After-tax risk-free rate \\
$\gamma$
  & Risk-aversion parameter
  & $\SR_W(\w)$
  & After-tax Sharpe ratio \\[6pt]
\midrule
\multicolumn{2}{l}{\textit{Statistical measures}}
  & \multicolumn{2}{l}{\textit{MM perspective (\Cref{sec:mm_perspective})}} \\
\midrule
$E[\cdot]$
  & Expectation
  & $V^0$, $V$
  & Firm value (pre-/post-tax) \\
$\SD(\cdot)$
  & Standard deviation
  & $\bar{x}$, $x_t$
  & Expected / realised cash flow \\
$\CV(\cdot)$
  & Coefficient of variation
  & $k$, $k^A$
  & Cost of capital (pre-/after-tax) \\
$\Cov(\cdot)$
  & Covariance
  & $\beta_U$
  & Unlevered (asset) beta \\
$\SR(\cdot)$
  & Sharpe ratio
  & $\beta_{\text{tax}}$
  & Beta of the tax claim \\[6pt]
\midrule
\multicolumn{2}{l}{\textit{CAPM (\Cref{sec:capm_specialcase})}}
  & \multicolumn{2}{l}{\textit{Beyond neutrality
    (\Cref{sec:beyond_neutrality})}} \\
\midrule
$\beta_j$, $\beta_j^A$
  & Asset beta; after-tax beta ($= \beta_j$)
  & $B$, $\theta$
  & Book value; book-to-market $B/V$ \\
$\text{MRP}$
  & Market risk premium $E[R_M] - \rf$
  & $V_{\textup{PE}}$, $V_{\textup{GE}}$
  & PE / GE asset value \\
$\text{MRP}^A$
  & After-tax MRP: $(1-\tw)\!\cdot\!\text{MRP}$
  & $c_j$
  & Illiquidity cost (stochastic) \\
$\beta_j^\tau$
  & After-tax beta (heterog.\ tax)
  & $\delta$
  & Payout ratio \\
$\tau_j$
  & Asset-specific effective tax rate
  & $\rho$
  & Internal rate of return \\
\bottomrule
\end{tabular}

\smallskip
\noindent\footnotesize
Subscript~$P$ denotes portfolio (pre-tax) quantities; subscript~$W$ denotes
after-tax wealth quantities. Superscripts~$A$ and~$B$ identify the taxed and
untaxed investor, respectively. Bold symbols ($\w$, $\R$, $\muvec$, $\V$,
$\Sig$, $\Z$) denote vectors or matrices; unbolded variants denote scalars.
\end{table}

% =========================================================================
\section{Single Asset Case Under GBM}\label{sec:single_gbm}

\subsection{Wealth Dynamics}\label{sec:gbm_wealth}

Consider a single risky asset with return parameters $(\mu, \sigma)$. The
investor holds $N_0$ shares at date~$0$ with initial price $P_0$, so initial
wealth is $W_0 = N_0 P_0$.

At each year-end $i = 1, 2, \ldots, n$:
\begin{enumerate}
  \item The share price evolves to $P_i$ according to the GBM.
  \item Dividends are paid and consumed (simplifying convention; see \Cref{app:dividends}).
  \item The wealth tax $\tw N_{i-1} P_i$ is due on the market value of holdings.
  \item The investor sells a fraction $\tw$ of shares to pay the tax.
\end{enumerate}

After $n$ years, the number of shares held is:
\begin{equation}\label{eq:shares_gbm}
  N_n = N_0 (1 - \tw)^n.
\end{equation}
The investor's equity wealth is:
\begin{equation}\label{eq:wealth_gbm}
  W_n = N_n P_n = W_0 (1 - \tw)^n \frac{P_n}{P_0}.
\end{equation}
Since $P_n / P_0$ is determined entirely by the GBM and is independent of the
investor's tax status, we can write:
\begin{equation}\label{eq:wealth_gbm_explicit}
  W_n = W_0 (1 - \tw)^n
    \exp\!\left[\left(\mu - \tfrac{1}{2}\sigma^2\right)n + \sigma Z_n\right]
\end{equation}
where $Z_n \sim N(0, n)$.

\subsection{Moments of the Wealth Distribution}\label{sec:gbm_moments}

\textbf{Expected wealth:}
\begin{equation}\label{eq:EW}
  E[W_n] = W_0 (1 - \tw)^n \, e^{\mu n}.
\end{equation}

\textbf{Standard deviation of wealth:}
\begin{equation}\label{eq:SDW}
  \SD(W_n) = W_0 (1 - \tw)^n \, e^{\mu n} \sqrt{e^{\sigma^2 n} - 1}.
\end{equation}

\textbf{Coefficient of variation:}
\begin{equation}\label{eq:CVW}
  \CV(W_n) = \frac{\SD(W_n)}{E[W_n]} = \sqrt{e^{\sigma^2 n} - 1}.
\end{equation}

\subsection{Key Result: Scaling Invariance}\label{sec:gbm_scaling}

The wealth tax enters only through the deterministic prefactor $(1 - \tw)^n$,
which multiplies both $E[W_n]$ and $\SD(W_n)$ identically. Therefore:

\begin{proposition}[CV Invariance, GBM]\label{prop:cv_gbm}
Under GBM with a proportional wealth tax on all assets, the coefficient of
variation of wealth at any horizon is invariant to the wealth tax rate:
\begin{equation}
  \CV_A(W_n) = \CV_B(W_n) = \sqrt{e^{\sigma^2 n} - 1}.
\end{equation}
The wealth tax reduces absolute expected wealth and absolute risk (standard
deviation) by the same proportion $(1 - \tw)^n$, leaving the relative
risk-reward profile unchanged.
\end{proposition}

This result generalises to any return distribution with finite second moments;
see \Cref{prop:cv_gen} in \Cref{sec:gen_cv}.

\subsection{Economic Interpretation: Proportional Dilution and Risk Sharing}\label{sec:gbm_liquidation}

The wealth tax is economically equivalent to an annual transfer of a
fraction $\tw$ of the investor's position to the government. Each share
retained has the same return distribution as it would in the absence of the
tax. The tax operates on the \emph{quantity} of the position, not on the
\emph{quality} of the return.

This equivalence is exact under the following conditions:
\begin{itemize}
  \item The tax rate $\tw$ is proportional (no exemptions, no progressive
        rates, no caps).
  \item The tax applies to all assets at the same rate.
  \item The asset price process $P_t$ is unaffected by the investor's tax
        status (partial equilibrium).
\end{itemize}

\begin{remark}[Risk sharing]
The proportional-dilution mechanism admits a dual interpretation.
From the Modigliani-Miller perspective developed in
\Cref{sec:mm_perspective}, the same mechanism can be viewed as the
government holding a proportional claim on the investor's assets. The
government participates in both upside and downside pro rata: the tax
reduces the investor's expected return and risk exposure by the same
factor~$(1-\tw)$. The wealth tax is therefore a form of risk sharing
between the investor and the state, not a one-sided extraction. These
two descriptions---proportional dilution (the investor's cash-flow
perspective) and proportional claim (the valuation perspective)---are
economically equivalent.
\end{remark}

\Cref{fig:proportional_dilution} illustrates the mechanism over several periods.

\begin{figure}[ht]
\centering
\resizebox{\textwidth}{!}{%
\begin{tikzpicture}[>=Stealth,
  yearbox/.style={draw, rounded corners=3pt, minimum width=2.0cm,
    minimum height=1.3cm, align=center, font=\small},
  arrowlabel/.style={midway, above, font=\footnotesize},
  taxlabel/.style={midway, below, font=\footnotesize, red!70!black}
]

% --- Year 0 ---
\node[yearbox, fill=blue!8] (y0) at (0,0)
  {$N_0$ shares\\[2pt]$W_0 = N_0 P_0$};
\node[font=\small\bfseries, above=4pt] at (y0.north) {Year 0};

% --- Year 1 ---
\node[yearbox, fill=blue!8] (y1) at (4.2,0)
  {$N_0(1-\tw)$\\[2pt]shares};
\node[font=\small\bfseries, above=4pt] at (y1.north) {Year 1};

% --- Year 2 ---
\node[yearbox, fill=blue!8] (y2) at (8.4,0)
  {$N_0(1-\tw)^2$\\[2pt]shares};
\node[font=\small\bfseries, above=4pt] at (y2.north) {Year 2};

% --- Year n ---
\node[yearbox, fill=blue!8] (yn) at (13.2,0)
  {$N_0(1-\tw)^n$\\[2pt]shares};
\node[font=\small\bfseries, above=4pt] at (yn.north) {Year $n$};

% --- Arrows ---
\draw[->, thick] (y0.east) -- (y1.west)
  node[arrowlabel] {$\times P_1/P_0$}
  node[taxlabel] {transfer $\tw N_0$};

\draw[->, thick] (y1.east) -- (y2.west)
  node[arrowlabel] {$\times P_2/P_1$}
  node[taxlabel] {transfer $\tw N_1$};

% --- Dots ---
\node[font=\Large] at (10.8,0) {$\cdots$};

% --- Bottom annotation ---
\node[anchor=north, font=\small, text width=13cm, align=center]
  at (6.6,-1.5) {%
  \textbf{Return per share:} $P_n/P_0$ \quad (identical for both investors)\\[3pt]
  \textbf{Wealth:} $W_n^A = \underbrace{(1-\tw)^n}_{\text{tax
  factor}} \times \underbrace{N_0 \cdot P_n}_{\text{untaxed wealth }W_n^B}$%
  };

\end{tikzpicture}%
}
\caption{The proportional-dilution mechanism. At each period end, a fraction
$\tw$ of the investor's shares is transferred to the government as tax. The
share price process $P_t$ is unaffected. After $n$~periods, the number of
shares has decayed by the deterministic factor $(1-\tw)^n$, while the return
per share remains $P_n/P_0$---identical for taxed and untaxed investors.}
\label{fig:proportional_dilution}
\end{figure}
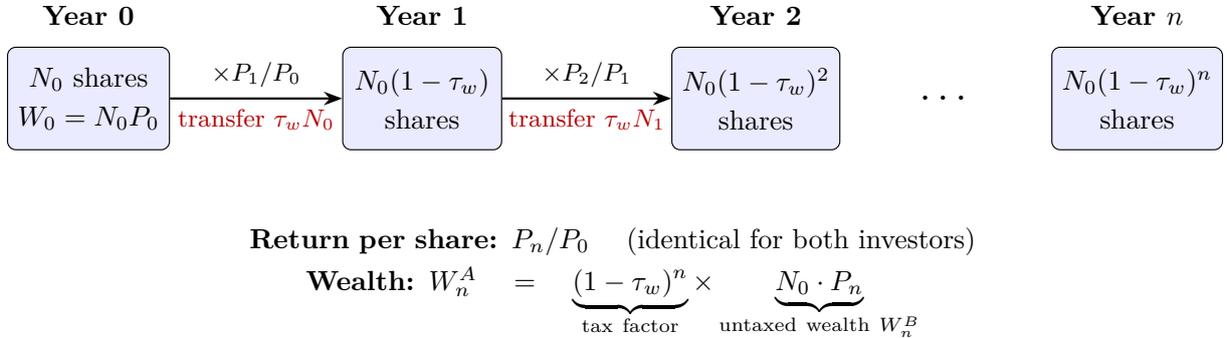

% =========================================================================
\section{Multi-Asset Portfolio Choice Under GBM}\label{sec:portfolio_gbm}

\subsection{Wealth Dynamics with Portfolio Weights}\label{sec:gbm_portfolio_dynamics}

Investor~B holds portfolio weights $\w$ in risky assets and
$(1 - \one^\top \w)$ in the risk-free asset. Following the continuous-time
framework of \citet{Merton1969}, wealth evolves as:
\begin{equation}\label{eq:dWB}
  \frac{dW_B}{W_B}
    = \bigl[\rf + \w^\top (\muvec - \rf \one)\bigr] dt
      + \w^\top \Sig \, d\Z.
\end{equation}
For Investor~A, the wealth tax applies to total wealth across all assets at
the same rate $\tw$. In the continuous-time approximation:
\begin{equation}\label{eq:dWA}
  \frac{dW_A}{W_A}
    = \bigl[\rf + \w^\top (\muvec - \rf \one) - \tw\bigr] dt
      + \w^\top \Sig \, d\Z.
\end{equation}
Define the portfolio expected return and volatility:
\begin{align}
  \mu_P(\w) &= \rf + \w^\top (\muvec - \rf \one), \label{eq:muP}\\
  \sigma_P(\w) &= \sqrt{\w^\top \V \w}. \label{eq:sigP}
\end{align}
Then the wealth dynamics decompose as:
\begin{align}
  \mu_W(\w, \tw) &= \mu_P(\w) - \tw, \label{eq:muW_ct}\\
  \sigma_W(\w, \tw) &= \sigma_P(\w). \label{eq:sigW_ct}
\end{align}
The wealth tax enters the drift additively and does not appear in the
diffusion coefficient.

\subsection{Optimality Conditions}\label{sec:gbm_foc}

The investor chooses $\w$ to maximise an objective $f(\mu_W, \sigma_W)$ that
is increasing in $\mu_W$ and decreasing in $\sigma_W$
\citep{Markowitz1952}. For concreteness, consider the mean-variance objective
with risk aversion parameter $\gamma$:
\begin{equation}\label{eq:mv_objective}
  \max_{\w} \quad \mu_P(\w) - \tw - \frac{\gamma}{2} \sigma_P^2(\w).
\end{equation}
The first-order condition is:
\begin{equation}
  \nabla_{\w} \mu_P = \gamma \, \V \w
  \qquad\Longrightarrow\qquad
  \muvec - \rf \one = \gamma \, \V \w^*.
\end{equation}
Solving:
\begin{equation}\label{eq:wstar_gbm}
  \w^* = \frac{1}{\gamma} \V^{-1} (\muvec - \rf \one).
\end{equation}

\begin{proposition}[Portfolio Invariance, GBM]\label{prop:portfolio_gbm}
The optimal portfolio weights $\w^*$ are independent of the wealth tax
rate~$\tw$:
\begin{equation}
  \w_A^* = \w_B^* = \frac{1}{\gamma} \V^{-1} (\muvec - \rf \one).
\end{equation}
This holds for any objective function $f(\mu_W, \sigma_W)$ where $\tw$ enters
only through an additive shift to $\mu_W$.
\end{proposition}

This full weight invariance is specific to the continuous-time formulation. In
discrete time, the tangency portfolio is still invariant, but the total risky
allocation may depend on preferences; see \Cref{prop:portfolio_gen} in
\Cref{sec:gen_portfolio} and the comparison in \Cref{app:ct_dt}.

\subsection{Orthogonality}\label{sec:gbm_orthogonality}

The invariance of $\w^*$ follows from a deeper structural property. The
gradients with respect to $\w$ are:
\begin{equation}
  \nabla_{\w} \mu_W = \muvec - \rf \one,
  \qquad
  \nabla_{\w} \sigma_W = \frac{\V \w}{\sigma_P}.
\end{equation}
Neither depends on $\tw$. The wealth tax shifts the objective in the
$\mu$-direction by a constant $-\tw$, which is orthogonal to the control
variable $\w$ that operates through both $\mu_P$ and $\sigma_P$.

\begin{proposition}[Orthogonality, GBM]\label{prop:orthogonality_gbm}
In the $(\sigma, \mu)$ plane, the wealth tax is a vertical translation of the
entire opportunity set (efficient frontier, capital allocation line, and
risk-free rate) by~$-\tw$. The indifference curves of any risk-averse investor
translate by the same amount. The tangency point---and hence the optimal
portfolio---moves purely vertically. The tax is orthogonal to portfolio choice.
\end{proposition}

Formally, define the direction of the wealth tax effect as
$\emu = (0, -1)$ in $(\sigma, \mu)$ space. Portfolio choice operates along the
efficient frontier, whose tangent direction has a nonzero $\sigma$-component at
any interior point. The inner product
$\emu \cdot \nabla_{\w} \sigma_P = 0$, confirming orthogonality.

\Cref{fig:orthogonality_ct} illustrates this geometry.

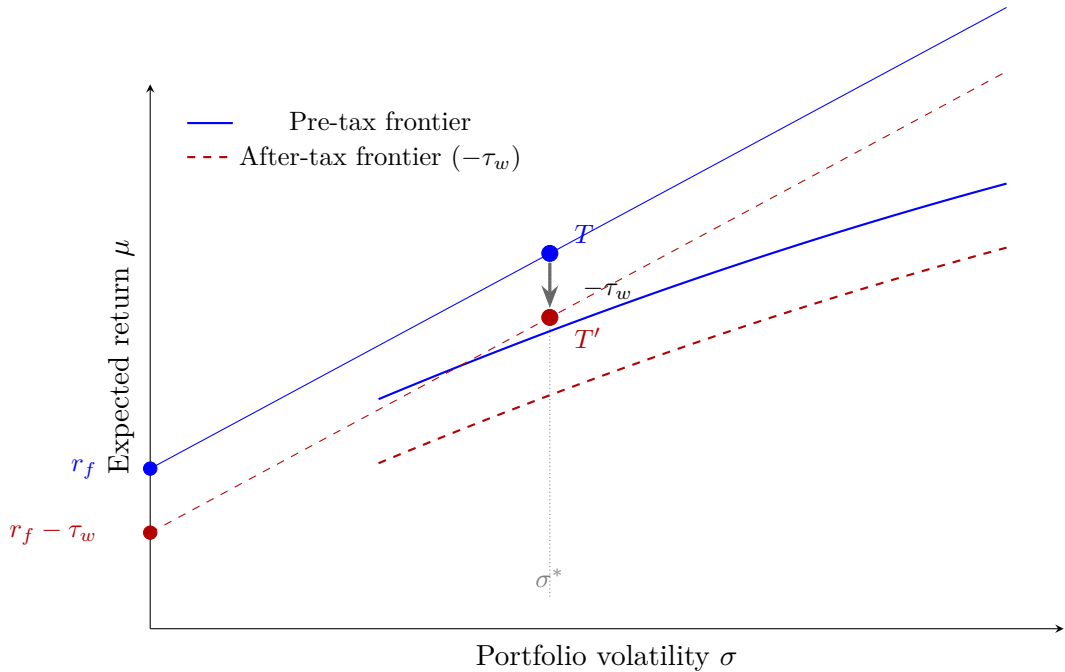
\begin{figure}[ht]
\centering
\begin{tikzpicture}[>=Stealth, scale=1]
\begin{axis}[
  width=0.85\textwidth, height=0.55\textwidth,
  xlabel={Portfolio volatility $\sigma$},
  ylabel={Expected return $\mu$},
  xmin=0, xmax=0.32, ymin=-0.01, ymax=0.16,
  xtick=\empty,
  ytick=\empty,
  axis lines=left,
  legend style={at={(0.03,0.97)}, anchor=north west, font=\small,
    draw=none, fill=white, fill opacity=0.8, text opacity=1},
  clip=false
]

% --- Pre-tax efficient frontier (hyperbola upper branch) ---
\addplot[blue, thick, domain=0.08:0.30, samples=80]
  ({x}, {0.03 + 0.42*x - 0.3*x^2});
\addlegendentry{Pre-tax frontier}

% --- After-tax efficient frontier (shifted down by tau_w = 0.02) ---
\addplot[red!70!black, thick, dashed, domain=0.08:0.30, samples=80]
  ({x}, {0.01 + 0.42*x - 0.3*x^2});
\addlegendentry{After-tax frontier ($-\tw$)}

% --- Pre-tax risk-free rate ---
\addplot[mark=*, mark size=2.5pt, blue, only marks] coordinates {(0, 0.04)};
\node[blue, anchor=east, font=\small] at (axis cs:-0.015,0.04) {$\rf$};

% --- After-tax risk-free rate ---
\addplot[mark=*, mark size=2.5pt, red!70!black, only marks] coordinates {(0, 0.02)};
\node[red!70!black, anchor=east, font=\small] at (axis cs:-0.015,0.02) {$\rf - \tw$};

% --- Pre-tax CAL ---
\addplot[blue, thin, domain=0:0.30, samples=2]
  ({x}, {0.04 + 0.48*x});

% --- After-tax CAL (same slope, lower intercept) ---
\addplot[red!70!black, thin, dashed, domain=0:0.30, samples=2]
  ({x}, {0.02 + 0.48*x});

% --- Tangency point pre-tax ---
\addplot[mark=*, mark size=3pt, blue, only marks] coordinates {(0.14, 0.1072)};
\node[blue, anchor=south west, font=\small] at (axis cs:0.145,0.1072) {$T$};

% --- Tangency point after-tax (same sigma, mu - tau_w) ---
\addplot[mark=*, mark size=3pt, red!70!black, only marks] coordinates {(0.14, 0.0872)};
\node[red!70!black, anchor=north west, font=\small] at (axis cs:0.145,0.0872) {$T'$};

% --- Vertical arrow showing the shift ---
\draw[->, thick, black!60, line width=1.2pt]
  (axis cs:0.14, 0.1042) -- (axis cs:0.14, 0.0902);
\node[anchor=west, font=\small] at (axis cs:0.148, 0.096) {$-\tw$};

% --- Dashed vertical line at sigma* ---
\draw[gray, thin, densely dotted]
  (axis cs:0.14, 0) -- (axis cs:0.14, 0.1072);
\node[gray, anchor=south, font=\small] at (axis cs:0.14, 0) {$\sigma^*$};

\end{axis}
\end{tikzpicture}
\caption{Orthogonality in continuous time. The wealth tax translates the
entire opportunity set vertically by~$-\tw$: the efficient frontier, the
risk-free rate, and the capital allocation line all shift down by the same
amount. The tangency portfolio~$T$ moves to~$T'$ at the same volatility
$\sigma^*$---the tax is orthogonal to portfolio choice.}
\label{fig:orthogonality_ct}
\end{figure}

In discrete time, the geometric interpretation changes from a vertical
translation to a homothetic contraction; see \Cref{prop:orthogonality_gen} in
\Cref{sec:gen_orthogonality} and the comparison in \Cref{app:ct_dt}.

% =========================================================================
\section{Asset Pricing Implications Under GBM}\label{sec:pricing_gbm}

\subsection{Per-Share Valuation}\label{sec:gbm_pershare}

Consider one share purchased at price $P_0$. After one year:
\begin{itemize}
  \item \textbf{Investor~B} holds 1~share worth $P_1$. Return: $P_1 / P_0$.
  \item \textbf{Investor~A} holds 1~share worth $P_1$, pays tax $\tw P_1$,
        retains $(1 - \tw)$ shares. Return per share held: still $P_1 / P_0$.
\end{itemize}
Investor~A receives this return on a shrinking number of shares, but the
return \emph{per share}---and hence the price they are willing to pay for any
individual share---is unchanged.

\begin{proposition}[Pricing Neutrality, GBM]\label{prop:pricing_gbm}
Under GBM with a proportional wealth tax on all assets, both investors are
willing to pay the same price per share for any asset. The wealth tax does not
create a pricing wedge.
\end{proposition}

This result---which parallels the DCF-based finding of
\citet{BjerksundSchjelderup2022}---is in fact completely distribution-free;
see \Cref{prop:pricing_gen} in \Cref{sec:gen_pricing}.

\subsection{Discussion}\label{sec:gbm_pricing_discussion}

This result may appear counterintuitive. The common argument is: ``Investor~A
faces an additional cost (the wealth tax), so they should demand a higher
return, which implies they would pay a lower price.'' The error in this
reasoning is the conflation of the return on the \emph{asset} with the return
on \emph{wealth}:
\begin{itemize}
  \item The \textbf{return on the asset} is $P_1 / P_0$, which is identical
        for both investors.
  \item The \textbf{return on wealth} is $P_1 / P_0$ for Investor~B and
        $(1 - \tw) P_1 / P_0$ for Investor~A.
\end{itemize}
The wealth tax reduces the return on wealth, not the return on the asset.
Since asset prices reflect asset returns---not any particular investor's wealth
accumulation rate---both investors value the asset identically.

The analogy is exact: the wealth tax is equivalent to any other proportional
personal expense (such as a fixed consumption rate proportional to wealth).
Such expenses reduce wealth accumulation but do not affect asset valuations.

% =========================================================================
\section{Summary of GBM Results}\label{sec:gbm_summary}

\begin{table}[ht]
\centering
\caption{Comparison of taxed and untaxed investors under GBM.}
\label{tab:gbm_summary}
\small
\begin{tabularx}{\textwidth}{@{}l X X@{}}
\toprule
Property & Investor~A (wealth tax $\tw$) & Investor~B (no tax) \\
\midrule
Return per share
  & $P_n / P_0$ & $P_n / P_0$ \\[4pt]
Shares held after $n$ years
  & $N_0 (1 - \tw)^n$ & $N_0$ \\[4pt]
Expected wealth $E[W_n]$
  & $W_0 (1-\tw)^n e^{\mu n}$ & $W_0 e^{\mu n}$ \\[4pt]
$\SD(W_n)$
  & $W_0 (1-\tw)^n e^{\mu n} \sqrt{e^{\sigma^2 n}-1}$
  & $W_0 e^{\mu n} \sqrt{e^{\sigma^2 n}-1}$ \\[4pt]
Coefficient of variation
  & $\sqrt{e^{\sigma^2 n}-1}$ & $\sqrt{e^{\sigma^2 n}-1}$ \\[4pt]
Optimal weights $\w^*$
  & $\frac{1}{\gamma}\V^{-1}(\muvec - \rf\one)$
  & $\frac{1}{\gamma}\V^{-1}(\muvec - \rf\one)$ \\[4pt]
Price per share & $P$ & $P$ \\[4pt]
Sharpe ratio
  & $\frac{\mu_T - \rf}{\sigma_T}$
  & $\frac{\mu_T - \rf}{\sigma_T}$ \\
\bottomrule
\end{tabularx}

\smallskip
\noindent\footnotesize
$\mu_T$ and $\sigma_T$ denote the tangent portfolio return and volatility.
\end{table}

% =========================================================================
\section{Generalisation Beyond Geometric Brownian Motion}\label{sec:generalisation}

\subsection{Motivation}\label{sec:gen_motivation}

The results in Sections~\ref{sec:single_gbm}--\ref{sec:pricing_gbm} were
derived under the assumption that asset prices follow a geometric Brownian
motion. While GBM provides a tractable and well-understood framework, it is a
strong assumption: it restricts returns to be lognormally distributed, rules
out fat tails, and imposes a continuous-time diffusion structure on prices.

We now show that none of these features are essential to the main results. The
four propositions depend not on the specific distributional form of returns,
but on the algebraic structure of how a proportional wealth tax interacts with
the return process. Specifically, the results rest on two pillars:
\begin{enumerate}
  \item \textbf{The wealth tax is a multiplicative scalar on wealth},
        deterministic and independent of the return realisation.
  \item \textbf{Asset returns have well-defined first and second moments}, so
        that the mean-variance framework applies.
\end{enumerate}
The first pillar is a consequence of proportionality and universality---the
same conditions already identified in the GBM analysis. The second replaces
the GBM assumption with a much weaker requirement that encompasses any
distribution in the \textbf{location-scale family} (normal, Student-$t$,
logistic, uniform, and more generally, the class of elliptical distributions).

We re-derive each proposition, clearly identifying the minimal assumptions
required. It will emerge that Propositions~$1'$ and~$4'$ require essentially no
distributional assumption at all, while Propositions~$2'$ and~$3'$ exploit the
location-scale structure to justify working in the $(\sigma, \mu)$ plane.

\subsection{Generalised Assumptions}\label{sec:gen_assumptions}

We retain the economic assumptions of \Cref{sec:setup} (proportional taxation,
universality, partial equilibrium, dividends consumed) and replace the GBM
return dynamics with the following:

\medskip
\noindent\textbf{Assumption A1 (Finite moments).}
The vector of asset returns $\R = (R_1, \ldots, R_K)^\top$ has a joint
distribution with well-defined mean vector $E[\R] = \muvec$ and
positive-definite covariance matrix $\Cov(\R) = \V$. No further assumption is
made about the form of the distribution.

\medskip
\noindent\textbf{Assumption A2 (Location-scale family).}
The return distribution belongs to the location-scale family: there exists a
standardised random vector $\Z$ with $E[\Z] = \mathbf{0}$ and
$\Cov(\Z) = \mathbf{I}$ such that
\begin{equation}\label{eq:loc_scale}
  \R \stackrel{d}{=} \muvec + \Sig \Z
\end{equation}
where $\Sig$ is a matrix with $\Sig\Sig^\top = \V$, and $\Z$ may follow any
distribution with well-defined second moments. The shape of the distribution
(e.g.\ tail behaviour, kurtosis) is encoded in $\Z$ and is invariant under
affine transformations.

\begin{remark}
Assumption~A2 is a strengthening of~A1 that is needed only for
Propositions~$2'$ and~$3'$. When A2 holds, the distribution is fully
characterised (up to shape) by its first two moments, which justifies the use
of mean-variance analysis \citep{Meyer1987,Chamberlain1983}. In the
multivariate setting, the natural extension is the class of \textbf{elliptical
distributions}, which nests the multivariate normal (and hence GBM) as a
special case \citep{OwenRabinovitch1983,HamadaValdez2008}. The precise role of
the distributional assumption is discussed in \Cref{app:distributional}.
\end{remark}

\noindent\textbf{Assumption A3 (Tax structure).}
Unchanged from \Cref{sec:setup_tax}: a proportional wealth tax at rate
$\tw \in (0,1)$ is levied on the market value of all assets (including the
risk-free asset) at each period end. The tax is paid by selling a fraction
$\tw$ of all positions.

\medskip
\noindent\textbf{Assumption A4 (Partial equilibrium).}
The return distribution $(\muvec, \V)$ is exogenous---it does not depend on
the tax rate.

\subsection{The Tax as a Multiplicative Scalar}\label{sec:gen_scalar}

The proportional-dilution interpretation from \Cref{sec:gbm_liquidation} is
entirely distribution-free. At each period end, a fraction
$\tw$ of all positions is transferred to the government as tax. After $n$ periods, the number of shares
held is:
\begin{equation}\label{eq:shares_gen}
  N_n = N_0(1 - \tw)^n
\end{equation}
and wealth is:
\begin{equation}\label{eq:wealth_gen}
  W_n^A = N_0(1 - \tw)^n P_n = (1 - \tw)^n \cdot W_0 \cdot \frac{P_n}{P_0}.
\end{equation}
Define the \textbf{cumulative gross return} $G^{(n)} \equiv P_n / P_0$. This
is a random variable whose distribution is determined by the asset return
process and is, by Assumption~A4, independent of the investor's tax status.
Then:
\begin{equation}\label{eq:mult_sep}
  W_n^A = (1 - \tw)^n \cdot W_0 \cdot G^{(n)} = (1 - \tw)^n \cdot W_n^B
\end{equation}
where $W_n^B = W_0 \cdot G^{(n)}$ is the wealth of the untaxed investor. The
wealth tax enters as a \textbf{deterministic multiplicative scalar}
$(1 - \tw)^n$ that is independent of the return realisation $G^{(n)}$.

This multiplicative separability is the structural engine behind all four
results. It requires only proportionality (the tax rate is constant) and
universality (all assets are taxed at the same rate). It does not require GBM,
continuity, or any specific distributional form.

Under GBM, the cumulative gross return takes the specific form
$G^{(n)} = \exp[(\mu - \tfrac{1}{2}\sigma^2)n + \sigma Z_n]$, but the
results below hold for any $G^{(n)}$ with well-defined moments.

\subsection{Generalised Proposition 1: CV Invariance}\label{sec:gen_cv}

\noindent\textbf{Requires:} A1 (finite moments), A3 (proportional tax). No
distributional assumption.

Since $W_n^A = (1 - \tw)^n \cdot W_n^B$ and $(1 - \tw)^n > 0$ is a
deterministic scalar:
\begin{align}
  E[W_n^A] &= (1 - \tw)^n \, E[W_n^B], \\
  \SD(W_n^A) &= (1 - \tw)^n \, \SD(W_n^B), \\
  \CV(W_n^A) &= \frac{\SD(W_n^A)}{E[W_n^A]}
              = \frac{\SD(W_n^B)}{E[W_n^B]} = \CV(W_n^B).
\end{align}

\begin{proposition}[Generalised CV Invariance]\label{prop:cv_gen}
Let asset returns have any joint distribution with well-defined first and
second moments (Assumption~A1). Under a proportional wealth tax on all assets
(Assumption~A3), the coefficient of variation of wealth at any horizon~$n$ is
invariant to the tax rate:
\begin{equation}
  \CV_A(W_n) = \CV_B(W_n) = \frac{\SD(G^{(n)})}{E[G^{(n)}]}
\end{equation}
where $G^{(n)} = P_n/P_0$ is the cumulative gross return. The result holds for
any return distribution---normal, lognormal, fat-tailed, skewed, discrete, or
continuous---provided the first two moments exist.
\end{proposition}

\begin{proof}
The proof is a single line of algebra exploiting the linearity of the
expectation operator and the absolute homogeneity of the standard deviation:
\begin{equation}
  \CV(cX) = \frac{|c| \, \SD(X)}{c \, E[X]} = \frac{\SD(X)}{E[X]} = \CV(X)
\end{equation}
for any positive constant $c$ and random variable $X$ with $E[X] > 0$. Setting
$c = (1 - \tw)^n$ and $X = W_n^B$ completes the argument.
\end{proof}

\begin{remark}
Under GBM, \Cref{prop:cv_gbm} gave the specific formula
$\CV(W_n) = \sqrt{e^{\sigma^2 n} - 1}$, which depended on the volatility
parameter $\sigma$ and the horizon $n$ through the lognormal moment structure.
\Cref{prop:cv_gen} shows that the \emph{invariance to~$\tw$} does not depend
on this formula---it is a consequence of the multiplicative structure of the
tax alone.
\end{remark}

\subsection{Generalised Proposition 2: Portfolio Weight Invariance}%
\label{sec:gen_portfolio}

\noindent\textbf{Requires:} A1 (finite moments), A3 (proportional tax). Full
result requires A2 (location-scale) or CRRA preferences.

\subsubsection{Setup}

Consider $K$ risky assets and a risk-free asset with one-period gross return
$(1 + \rf)$. The investor allocates portfolio weights
$\w = (w_1, \ldots, w_K)^\top$ to risky assets, with the remainder
$(1 - \one^\top \w)$ in the risk-free asset. The portfolio rate of return is:
\begin{equation}\label{eq:RP_gen}
  R_P(\w) = \rf + \w^\top(\R - \rf \one)
\end{equation}
with mean and variance:
\begin{equation}\label{eq:muP_sigP_gen}
  \mu_P(\w) = \rf + \w^\top(\muvec - \rf \one),
  \qquad
  \sigma_P^2(\w) = \w^\top \V \w.
\end{equation}
This notation is consistent with the portfolio return and volatility defined in
\Cref{sec:gbm_portfolio_dynamics}, but now $\R$ may follow any distribution
satisfying~A1.

\subsubsection{After-Tax Returns}

After paying the wealth tax, end-of-period wealth is:
\begin{equation}\label{eq:W1A_gen}
  W_1^A = (1 - \tw) \cdot W_0(1 + R_P) = (1 - \tw) \cdot W_1^B.
\end{equation}
The after-tax rate of return (on beginning-of-period wealth) is:
\begin{equation}\label{eq:RW_gen}
  R_W(\w, \tw) = (1 - \tw)(1 + R_P(\w)) - 1
\end{equation}
with moments:
\begin{align}
  \mu_W(\w, \tw) &= (1 - \tw)(1 + \mu_P(\w)) - 1, \label{eq:muW_gen}\\
  \sigma_W(\w, \tw) &= (1 - \tw) \, \sigma_P(\w). \label{eq:sigW_gen}
\end{align}
The after-tax risk-free return (setting $\w = \mathbf{0}$) is:
\begin{equation}\label{eq:rfA}
  \rf^A = (1 - \tw)(1 + \rf) - 1 = \rf - \tw(1 + \rf).
\end{equation}
Note the contrast with the continuous-time GBM case
(\Cref{sec:gbm_portfolio_dynamics}), where the tax entered the drift
additively ($\mu_W = \mu_P - \tw$) and did not affect volatility
($\sigma_W = \sigma_P$). In the discrete-time formulation, the tax scales both
the mean and the standard deviation by $(1 - \tw)$. This distinction is
important for the portfolio choice result but does not affect the Sharpe ratio,
as we now show.

\subsubsection{Sharpe Ratio Preservation}

The after-tax excess return of portfolio $\w$ over the after-tax risk-free rate
is:
\begin{equation}
  \mu_W - \rf^A = (1 - \tw)(1 + \mu_P) - 1 - [\rf - \tw(1 + \rf)]
               = (1 - \tw)(\mu_P - \rf).
\end{equation}
The after-tax Sharpe ratio \citep{Sharpe1966} is therefore:
\begin{equation}\label{eq:SR_preservation}
  \SR_W(\w) = \frac{\mu_W - \rf^A}{\sigma_W}
           = \frac{(1 - \tw)(\mu_P - \rf)}{(1 - \tw)\sigma_P}
           = \frac{\mu_P - \rf}{\sigma_P}
           = \SR_P(\w).
\end{equation}
The $(1 - \tw)$ factors cancel exactly. The Sharpe ratio of every portfolio is
invariant to the wealth tax.

\subsubsection{Tangency Portfolio Invariance}

The tangency portfolio $\w_T$ is the portfolio that maximises the Sharpe ratio:
\begin{equation}
  \w_T = \argmax_{\w} \frac{\mu_P(\w) - \rf}{\sigma_P(\w)}
       = \argmax_{\w} \SR_P(\w).
\end{equation}
Since $\SR_W(\w) = \SR_P(\w)$ for all $\w$, the tangency portfolio is the same
in the after-tax space as in the pre-tax space. The standard solution is:
\begin{equation}\label{eq:wT}
  \w_T \propto \V^{-1}(\muvec - \rf \one)
\end{equation}
which does not depend on $\tw$.

\subsubsection{Full Weight Invariance Under CRRA}

For the allocation between the tangency portfolio and the risk-free asset, we
need to specify preferences. Under CRRA utility
$U(W) = W^{1-\gamma}/(1-\gamma)$:
\begin{equation}
  E[U(W_1^A)] = \frac{[(1 - \tw) W_0]^{1-\gamma}}{1 - \gamma}
    \, E\!\left[(1 + R_P(\w))^{1-\gamma}\right].
\end{equation}
The prefactor $[(1 - \tw) W_0]^{1-\gamma}$ is a positive constant that does
not affect the $\argmax$ over $\w$. For log utility ($\gamma = 1$):
\begin{equation}
  E[\log W_1^A] = \log[(1 - \tw) W_0] + E[\log(1 + R_P(\w))].
\end{equation}
Again, the first term is constant. Therefore:

\begin{proposition}[Generalised Portfolio Invariance]\label{prop:portfolio_gen}
Let asset returns have any joint distribution with well-defined first and
second moments (Assumption~A1), and let the wealth tax be proportional on all
assets (Assumption~A3).

\noindent (a) The Sharpe ratio of every portfolio is invariant to the wealth
tax rate:
\begin{equation}
  \SR_W(\w) = \SR_P(\w) \qquad \forall \, \w.
\end{equation}
The tangency portfolio---the optimal allocation among risky assets---is
identical for both investors.

\noindent (b) Under CRRA preferences (including log utility), the full optimal
weight vector $\w^*$ (including the allocation between risky and risk-free
assets) is independent of $\tw$:
\begin{equation}
  \w_A^* = \w_B^*.
\end{equation}
These results hold without any specific distributional assumption on returns.
\end{proposition}

\begin{remark}[Mean-variance preferences]
For a mean-variance investor maximising
$\mu_W - \frac{\gamma}{2}\sigma_W^2$ over after-tax returns, the first-order
condition yields
$\w^* = \frac{1}{\gamma(1-\tw)}\V^{-1}(\muvec - \rf\one)$. The tax affects
the total risky allocation through the factor $1/(1-\tw)$, but the
\emph{composition} of the risky portfolio (the tangency portfolio) is
unchanged. In the continuous-time limit, this scaling effect vanishes and the
full weight vector is invariant, recovering \Cref{prop:portfolio_gbm}. See
\Cref{app:ct_dt} for a systematic comparison.
\end{remark}

\subsection{Generalised Proposition 3: Orthogonality}\label{sec:gen_orthogonality}

\noindent\textbf{Requires:} A2 (location-scale family), to justify the
$(\sigma, \mu)$ representation.

\subsubsection{The Discrete-Time Geometry}

Under the location-scale assumption (A2), any investor's feasible set can be
represented in the mean-standard deviation plane. The efficient frontier, the
capital allocation line, and the indifference curves all live in this plane.

Consider the transformation induced by the wealth tax on after-tax returns.
Writing in terms of the \textbf{gross return} $G_P = 1 + R_P$ and the
after-tax gross return $G_W = (1 - \tw) G_P$:
\begin{equation}
  E[G_W] = (1 - \tw) \, E[G_P],
  \qquad
  \SD(G_W) = (1 - \tw) \, \SD(G_P).
\end{equation}
In the $(\sigma, G)$ plane, the wealth tax maps every point
$(\sigma_P, \bar{G}_P)$ to
$((1-\tw)\sigma_P, \, (1-\tw)\bar{G}_P)$. This is a \textbf{homothetic
contraction} centred at the origin, with factor $(1 - \tw) < 1$.

\subsubsection{Properties of the Contraction}

The homothetic contraction has several important geometric properties:
\begin{enumerate}
  \item \textbf{Every point moves radially inward toward the origin.} The
        direction from the origin to any point is preserved; only the distance
        changes.
  \item \textbf{The slope of every ray through the origin is preserved.} In
        particular, the slope of the capital allocation line (which is the
        Sharpe ratio plus a constant related to $\rf$) is invariant.
  \item \textbf{The efficient frontier contracts uniformly.} Every point on
        the frontier maps to a point on the after-tax frontier at the same
        angular position relative to the origin.
  \item \textbf{The tangency point is preserved.} Since both the frontier and
        the risk-free point contract along their respective rays, and the slope
        of the line connecting them is preserved, the tangency occurs at the
        same portfolio.
\end{enumerate}

\subsubsection{Comparison with the Continuous-Time Case}

In continuous time (\Cref{sec:gbm_orthogonality}), the wealth tax induced a
\textbf{pure vertical translation} of the entire $(\sigma, \mu)$ plane by
$-\tw$: the efficient frontier shifted down, the risk-free rate shifted down,
and the tangency point moved straight down. The tax direction
$\emu = (0, -1)$ was literally orthogonal to the frontier's tangent, which had
a nonzero $\sigma$-component.

In discrete time, the transformation is richer: it is a contraction (scaling of
both axes), not a translation (shift of one axis). Nevertheless, the economic
content is the same:

\begin{proposition}[Generalised Orthogonality]\label{prop:orthogonality_gen}
Under the location-scale assumption (A2), the wealth tax induces a homothetic
contraction of the opportunity set in the $(\sigma, \, 1+\mu)$ plane by the
factor $(1 - \tw)$. The contraction preserves:
\begin{itemize}
  \item the Sharpe ratio of every portfolio,
  \item the composition of the tangency portfolio,
  \item the slope of the capital allocation line.
\end{itemize}
The tax operates purely as a radial scaling of the opportunity set. Portfolio
optimisation selects the tangency point based on the angular position of the
frontier relative to the risk-free ray---a quantity that is invariant under
radial scaling. In this sense, the wealth tax remains orthogonal to portfolio
choice.
\end{proposition}

In the continuous-time limit ($\tw \, dt \to 0$), the contraction degenerates
into a vertical translation, recovering the exact orthogonality of
\Cref{prop:orthogonality_gbm}. \Cref{fig:contraction_dt} illustrates the
discrete-time geometry; compare with the continuous-time case in
\Cref{fig:orthogonality_ct}.

\begin{figure}[ht]
\centering
\begin{tikzpicture}[>=Stealth, scale=1]
\begin{axis}[
  width=0.85\textwidth, height=0.6\textwidth,
  xlabel={$\SD(G_P)$},
  ylabel={$E[G_P] = 1 + \mu_P$},
  xmin=0, xmax=0.32, ymin=0, ymax=1.25,
  xtick=\empty,
  ytick=\empty,
  axis lines=left,
  clip=false
]

% --- Pre-tax efficient frontier ---
\addplot[blue, thick, domain=0.06:0.28, samples=80]
  ({x}, {1.03 + 0.55*x - 0.4*x^2});

% --- After-tax efficient frontier (homothetic contraction, factor 0.85) ---
\addplot[red!70!black, thick, dashed, domain=0.06:0.28, samples=80]
  ({0.85*x}, {0.85*(1.03 + 0.55*x - 0.4*x^2)});

% --- Pre-tax CAL ---
\addplot[blue, thin, domain=0:0.28, samples=2]
  ({x}, {1.04 + 0.55*x});

% --- After-tax CAL (same slope, lower intercept) ---
\addplot[red!70!black, thin, dashed, domain=0:0.24, samples=2]
  ({x}, {0.85*1.04 + 0.55*x});

% --- Pre-tax risk-free gross return ---
\addplot[mark=*, mark size=2.5pt, blue, only marks] coordinates {(0, 1.04)};

% --- After-tax risk-free gross return ---
\addplot[mark=*, mark size=2.5pt, red!70!black, only marks]
  coordinates {(0, 0.884)};

% --- Risk-free labels (placed well left of axis) ---
\node[blue, anchor=east, font=\small] at (axis cs:-0.015, 1.04)
  {$1+\rf$};
\node[red!70!black, anchor=east, font=\small] at (axis cs:-0.015, 0.884)
  {$(1{-}\tw)(1{+}\rf)$};

% --- Tangency point pre-tax ---
\addplot[mark=*, mark size=3pt, blue, only marks] coordinates {(0.14, 1.1095)};

% --- Tangency point after-tax ---
\addplot[mark=*, mark size=3pt, red!70!black, only marks]
  coordinates {(0.119, 0.943)};

% --- Tangency labels (placed away from curves) ---
\node[blue, anchor=south, font=\small] at (axis cs:0.14, 1.12) {$T$};
\node[red!70!black, anchor=north, font=\small] at (axis cs:0.119, 0.93) {$T'$};

% --- Ray from origin through tangency ---
\draw[gray, thin, densely dotted] (axis cs:0,0) -- (axis cs:0.158, 1.25);

% --- Arrow showing radial contraction ---
\draw[->, thick, black!60, line width=1.2pt]
  (axis cs:0.136, 1.075) -- (axis cs:0.122, 0.965);
\node[anchor=west, font=\small] at (axis cs:0.145, 1.02) {$(1{-}\tw)$};

% --- Origin ---
\addplot[mark=*, mark size=1.5pt, black, only marks] coordinates {(0, 0)};

% --- Legend (placed underneath respective curves) ---
\node[anchor=north, font=\small, text=blue] at (axis cs:0.22, 1.06)
  {Pre-tax frontier};
\node[anchor=north, font=\small, text=red!70!black] at (axis cs:0.20, 0.88)
  {After-tax frontier};

\end{axis}
\end{tikzpicture}
\caption{Orthogonality in discrete time. In the $(\sigma, \, 1{+}\mu)$ plane,
the wealth tax induces a homothetic contraction of the opportunity set toward
the origin by the factor $(1-\tw)$. Every point moves radially inward: the
efficient frontier, the risk-free gross return, and the tangency portfolio~$T$
all contract along their respective rays from the origin. The slope of the
capital allocation line---and hence the Sharpe ratio---is preserved.}
\label{fig:contraction_dt}
\end{figure}
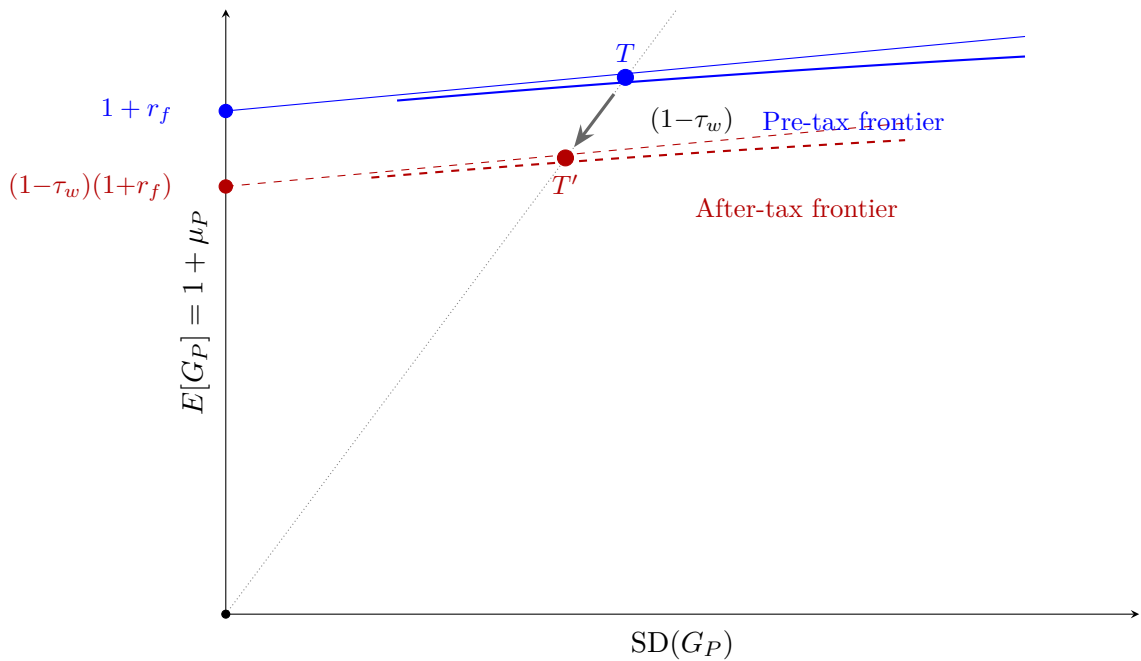

\subsection{Generalised Proposition 4: Pricing Neutrality}\label{sec:gen_pricing}

\noindent\textbf{Requires:} A3 (proportional tax), A4 (partial equilibrium).
No distributional assumption.

The pricing result is the most general of the four. The argument of
\Cref{sec:gbm_pershare} is entirely distribution-free: the return per share is
$P_1/P_0$ regardless of the investor's tax status, because the tax operates on
the \emph{number} of shares, not on the \emph{return per share}. No assumption
about the distribution of $P_1$ is needed.

For discounted cash flow valuation, this result is confirmed by
\citet{BjerksundSchjelderup2022}, who show that the domestic (taxed) investor's
NPV equals that of the foreign (untaxed) investor. The mechanism is discussed
in \Cref{sec:independence_capm}.

\begin{proposition}[Generalised Pricing Neutrality]\label{prop:pricing_gen}
Under a proportional wealth tax on all assets (Assumption~A3) and exogenous
asset return distributions (Assumption~A4), both investors are willing to pay
the same price per share for any asset. This result is completely
distribution-free: it does not require GBM, finite moments, or any specific
return process.
\end{proposition}

\begin{remark}[Tax base timing]\label{rem:tax_base_timing}
In the model, the wealth tax is levied at period end on the current
market value of holdings (Assumption~A3).  The neutrality result survives
because the tax operates as a fixed fractional sale of shares: the
investor sells a fraction~$\tw$ of all positions regardless of the
realised price, so the multiplicative structure $(1-\tw)^n$ is
deterministic and separates cleanly from the stochastic return.
\citet{Kruschwitz2023} show that using end-of-period market values can
create arbitrage in a cost-of-capital framework that sets
$k^\tau = k - \tw$.  Our after-tax discount rate
$k^A = (1-\tw)(1+k) - 1$ (equation~\eqref{eq:mm_kA}) differs from the
naive $k - \tw$ by a term~$-\tw k$, and is consistent with no-arbitrage
under end-of-period assessment.  See \Cref{sec:conditions_economic} for
a fuller discussion and the connection to Norwegian institutional
practice.
\end{remark}

The role of distributional assumptions---which results require
the location-scale family and which are distribution-free---is discussed
in \Cref{app:distributional}.

% =========================================================================
\section{Discussion}%
\label{sec:conditions}

\subsection{Essential Economic Conditions}\label{sec:conditions_economic}

The results in Sections~\ref{sec:single_gbm}--\ref{sec:generalisation} hold
under specific economic assumptions. Each assumption, if relaxed, opens a
channel through which the wealth tax \emph{does} affect the risk-reward profile
or portfolio choice:

\textbf{Proportionality.} The tax must be a fixed fraction of market value.
Exemption thresholds, progressive rates, or caps introduce nonlinear
transformations that change the shape of the wealth distribution and break CV
invariance.

\textbf{Universality.} The tax must apply to all assets at the same rate. If
some assets are exempt (e.g., the risk-free asset, retirement accounts, real
estate), the tax creates differential costs across assets, which distorts
relative Sharpe ratios and portfolio weights.

\textbf{Tax base specification.} The results assume the wealth tax base is the
market value of holdings.  Two aspects of the tax base matter: the
\emph{valuation basis} (market value vs.\ book value, analysed in
\Cref{sec:bv_taxation}) and the \emph{timing} (beginning vs.\ end of
period).

\citet{Kruschwitz2023} show that in a cost-of-capital framework, using
end-of-period market values with the naive discount rate
$k^\tau = k - \tw$ creates arbitrage opportunities.  Our model levies
the tax at period end on current market values (Assumption~A3), but
avoids this problem for two reasons.  First, the proportional-dilution
mechanism---transferring a fixed fraction~$\tw$ of shares regardless of the
realised price---generates the deterministic factor $(1-\tw)^n$ that
separates from the stochastic return without requiring a modified
discount rate.  Second, where a discount rate is needed (the MM
perspective, \Cref{sec:mm_perspective}), we derive
$k^A = (1-\tw)(1+k) - 1$, which differs from $k - \tw$ by the
correction term~$-\tw k$ and is arbitrage-free under end-of-period
assessment.

The alternative arbitrage-free discount rate derived by
\citet{Kruschwitz2023} for end-of-period assessment,
$k^K = k - \tw(1+k)/(1+\rf)$, rests on a narrower tax base.  In their
pricing equation, $(1-\tw)$ multiplies only the continuation market
value~$V_{t+1}^\tau$, not the cash flow~$\text{CF}_{t+1}$: the wealth
tax is levied on the ex-dividend market value of the security, while the
dividend escapes the tax base entirely.  Their formula therefore requires
that no part of the dividend remains in the investor's taxable estate at
the assessment date---an implicit assumption that is difficult to reconcile
with a comprehensive wealth tax on net worth.  If dividends are instead
retained (as cash, deposits, or reinvested capital) and therefore remain
within the tax base, then $(1-\tw)$ applies to the entire gross payoff and
our multiplicative formula obtains (see \Cref{app:dividends} for a detailed
analysis of alternative payment mechanisms).

\begin{remark}[Consistency of the Kruschwitz et al.\ setup]
The arbitrage identified by \citet{Kruschwitz2023} is in fact present in
their model even at $\tw = 0$, before any tax is introduced. Their
two-state binomial setup specifies three securities (one bond and two risky
assets) with independently chosen costs of capital $k_1$ and~$k_2$.
In a two-state world, however, no-arbitrage constrains the risk premia:
\begin{equation}\label{eq:kruschwitz_noarb}
  \frac{k_1 - r_f}{1 + k_1}
  = \frac{u_1}{u_2} \cdot \frac{k_2 - r_f}{1 + k_2},
\end{equation}
where $u_1/u_2$ is the ratio of up-state payoffs. For the numerical
parameters used by Kruschwitz et al., this condition fails. Setting their
arbitrage-profit expression (their equation~9) to $\tw = 0$ yields a
non-zero profit, confirming that the overcomplete market is already
mispriced. The arbitrage they attribute to the naive discount rate
$k^\tau = k - \tw$ is therefore an artefact of a pre-existing model
inconsistency, not a consequence of wealth taxation.  We thank Petter
Bjerksund for drawing our attention to this observation.
\end{remark}

\citet{BjerksundSchjelderup2022} assume the tax base is the
beginning-of-period market value, for which $k - \tw$ is the correct
arbitrage-free discount rate \citep{Kruschwitz2023}.  In the Norwegian
institutional setting, both timing conventions coexist.  Listed shares
held \emph{directly} are assessed at market value on 1~January of the
assessment year---effectively the end-of-period price
(\emph{skatteloven}~\S~4-12(1), \S~4-1).  However, most substantial
Norwegian investors hold listed equities through unlisted holding
companies under the participation exemption
(\emph{fritaksmetoden}).  The holding company's shares are assessed at
tax book value per 1~January of the \emph{income} year---the beginning
of the period---so the effective tax base for these investors is a
lagged, predetermined quantity, consistent with the
\citeauthor{BjerksundSchjelderup2022} assumption.  For unlisted
operating companies, the tax base is the company's net asset value at
1~January of the income year, which is the book-value case analysed in
\Cref{sec:bv_taxation}.  The distinction between end-of-period and
beginning-of-period assessment is thus primarily relevant for
\emph{direct} holdings of listed shares.

Our framework nests both cases.  End-of-period assessment is the general
case analysed throughout the paper, with after-tax discount rate
$k^A = (1-\tw)(1+k) - 1 = k - \tw - \tw k$.  Beginning-of-period
assessment is the special case in which the tax base is predetermined
(known at the start of the period).  When the tax base is deterministic,
the cross term~$\tw k$ vanishes from the pricing equation, and $k^A$
reduces to $k - \tw$---the \citeauthor{BjerksundSchjelderup2022} rate.
In both cases, the tax reduces the numerator and denominator of the
valuation in the same proportion, and pricing neutrality ($V = V^0$)
obtains.

\textbf{Partial equilibrium.} The asset price process $(\muvec, \V)$ is taken
as given. In general equilibrium where all investors face the wealth tax, the
return process may in principle adjust. However, under CRRA preferences the
optimal portfolio weights are independent of~$\tw$
(\Cref{prop:portfolio_gen}b), so that market clearing produces the same
equilibrium returns and prices---the partial-equilibrium assumption is
self-fulfilling (\Cref{sec:capm_pe_ge}). Departure from partial equilibrium
therefore requires either non-homothetic preferences or a violation of the
other conditions listed here (universality, market-value tax base).

\textbf{No borrowing constraints.} If the wealth tax reduces the effective
return sufficiently that investors wish to lever up, binding leverage
constraints create differential effects across wealth levels.

\textbf{All wealth remains within the tax base.} The main text assumes
dividends are consumed and the tax is paid by selling shares, which yields
the clean proportional-dilution formula $N_n = N_0(1-\tw)^n$.  This is a
simplifying convention.  What matters for neutrality is that all wealth---
including reinvested dividends---remains subject to wealth tax at market
value.  If dividends are reinvested in taxable assets, the multiplicative
separability $W_n^A = (1-\tw)^n W_n^B$ is preserved, because the tax factor
$(1-\tw)$ applies to total wealth regardless of its composition
(\Cref{app:dividends}).  Neutrality would break if dividends could be
channelled into assets that escape the tax base---but that is a violation of
universality, not a consequence of reinvestment per se.

\textbf{Frictionless rebalancing.} The multi-asset portfolio results
(Propositions~\ref{prop:portfolio_gbm}--\ref{prop:portfolio_gen}) assume
that the investor can rebalance to the optimal weights~$\w^*$ at each period
end without transaction costs. If rebalancing incurred proportional costs
(bid-ask spreads, commissions, or market impact), these costs would reduce
wealth in a way that depends on portfolio composition, the degree of weight
drift, and the volatility of relative asset returns---introducing a channel
through which the wealth tax could interact with portfolio choice indirectly,
since a lower wealth base may alter the cost-benefit calculus of
rebalancing. In the single-asset case, no rebalancing is needed and this
condition is vacuous.

\textbf{Tax payment mechanism.} The main text assumes the tax is paid by
selling shares (proportional dilution). If instead the tax is paid from
dividend income, the results are preserved when the dividend yield is
deterministic, but the multiplicative separability can break down when the
dividend yield is stochastic. See \Cref{app:dividends} for a detailed
analysis.

\medskip
\noindent
\Cref{sec:beyond_neutrality} formalises three channels through which
relaxing these conditions produces non-neutral effects: book-value taxation
(\Cref{sec:bv_taxation}), liquidity frictions (\Cref{sec:liq_frictions}),
and dividend extraction (\Cref{sec:div_extraction}).

\subsection{Independence from Asset Pricing Models}\label{sec:independence_capm}

A common source of confusion in discussions of wealth taxation and asset
valuation is the role of asset pricing models---in particular, whether results
such as those presented here presuppose the Capital Asset Pricing Model (CAPM)
or some other equilibrium pricing framework. We address this explicitly,
because the distinction between portfolio theory and equilibrium pricing is
frequently conflated.

\textbf{Our results do not assume CAPM.} The propositions in this paper are
derived from \textbf{portfolio theory} \citep{Markowitz1952}---the investor's
demand-side optimisation problem. We take the return distribution
$(\muvec, \V)$ as exogenous (Assumption~A4) and solve for the optimal
portfolio weights. This is not an equilibrium statement: we never derive
$\muvec$ from a market-clearing condition or from a factor model. The expected
return vector could come from CAPM, from an APT model, from a multi-factor
model, or simply from historical estimation---the results are identical in all
cases. What matters is that $(\muvec, \V)$ exists and does not depend on the
tax rate; how it is determined is irrelevant to the analysis.

The distinction is important: CAPM is a statement about equilibrium asset
prices (it says the expected excess return on any asset is proportional to its
beta with the market portfolio). Our propositions are statements about an
individual investor's portfolio problem, conditional on whatever expected
returns the market offers. These are logically independent.

\textbf{Bjerksund and Schjelderup (2022) do not depend on CAPM.} Their pricing
neutrality result (which corresponds to our \Cref{prop:pricing_gen}) is
sometimes misread as requiring CAPM, perhaps because they refer to a ``cost of
capital'' $k$ and to ``relevant risk characteristics'' of the asset. In fact,
their derivation rests on a \textbf{no-arbitrage / efficient markets}
assumption: in an efficient market, the investor's opportunity return (the
discount rate) equals the expected return on an investment with the same risk
characteristics as the asset. This is a no-arbitrage condition, not a
CAPM-specific one.

CAPM appears in \citet{BjerksundSchjelderup2022} only in footnotes---footnote~6
(p.~876) mentions the CAPM formula
$r = \rf + \beta(r_M - \rf)$ as one example of how the expected
return~$r$ could be explained, and footnote~7 (p.~877) notes that ``within the
capital asset pricing model, for instance, the relevant risk characteristic is
measured by beta.'' Both statements are illustrative, not assumptions. The
mechanism behind pricing neutrality---that the tax simultaneously reduces the
expected cash flows and the discount rate, and these effects cancel
exactly---works for \textbf{any} pricing model that provides a discount rate
consistent with no-arbitrage.

\textbf{Summary of what is assumed.} \Cref{tab:frameworks} clarifies the
pricing framework used at each level.

\begin{table}[ht]
\centering
\caption{Pricing frameworks and their role in this paper.}
\label{tab:frameworks}
\small
\begin{tabular}{@{}llll@{}}
\toprule
Framework & Provides & Assumes & Used here? \\
\midrule
No-arbitrage & Discount rate $=$ opp.\ return
  & Efficient markets & Yes (Prop $4'$, B\&S) \\
Markowitz & Optimal portfolio weights
  & $(\muvec,\V)$ exist; risk aversion & Yes (Props $1'$--$3'$) \\
CAPM & $\muvec$ from equilibrium
  & Homogeneous expect. & \textbf{No} \\
Multi-factor / APT & $\muvec$ from factor loadings
  & Factor structure & \textbf{No} \\
\bottomrule
\end{tabular}
\end{table}

This distinction is a strength of the analysis: the results hold for any
pricing model consistent with no-arbitrage, including CAPM, but they do not
require it. The compatibility of our distributional assumptions with CAPM (via
the elliptical class) is discussed in \Cref{app:distributional} and in the entry on
\citet{HamadaValdez2008} in \Cref{app:literature}.

\subsection{The CAPM Special Case}\label{sec:capm_specialcase}

The preceding subsection established that our results do not require CAPM.
We now show what happens when CAPM \emph{is} imposed as the equilibrium
pricing model. This exercise provides three additional insights: it confirms
that the partial-equilibrium results extend to general equilibrium under
CRRA preferences, sharpening the caveat in \Cref{sec:conditions_economic};
it identifies a precise error in \citet{Fama2021}; and it provides the
equilibrium mechanism for the beta-dependent pricing in
\Cref{sec:bv_multiperiod}.

\subsubsection{After-tax beta invariance}

Under the proportional wealth tax, the after-tax return on asset~$j$ is
$R_j^A = (1 - \tw)(1 + R_j) - 1$ (\Cref{eq:RW_gen}) and the after-tax
market return is $R_M^A = (1 - \tw)(1 + R_M) - 1$. The after-tax beta is:
\begin{equation}\label{eq:beta_aftertax}
  \beta_j^A \;=\; \frac{\Cov(R_j^A,\, R_M^A)}{\Var(R_M^A)}
             \;=\; \frac{(1-\tw)^2\,\Cov(R_j,\, R_M)}
                        {(1-\tw)^2\,\Var(R_M)}
             \;=\; \beta_j.
\end{equation}
The $(1-\tw)^2$ cancels: after-tax beta equals pre-tax beta. This is the
CAPM counterpart of the tangency-portfolio invariance in
\Cref{prop:portfolio_gen}.

\subsubsection{The after-tax security market line}\label{sec:capm_sml}

Investor~B (untaxed) faces the standard SML:
\begin{equation}\label{eq:sml_pretax}
  E[R_j] \;=\; \rf + \beta_j \cdot
  \underbrace{\bigl(E[R_M] - \rf\bigr)}_{\text{MRP}}.
\end{equation}
Investor~A (taxed) faces the after-tax SML:
\begin{equation}\label{eq:sml_aftertax}
  E[R_j^A] \;=\; \rf^A + \beta_j \cdot
  \underbrace{\bigl(E[R_M^A] - \rf^A\bigr)}_{\text{MRP}^A}
\end{equation}
where, from \eqref{eq:rfA} and \eqref{eq:muW_gen},
\begin{equation}\label{eq:mrp_aftertax}
  \text{MRP}^A \;=\; (1-\tw)\bigl(E[R_M] - \rf\bigr)
               \;=\; (1-\tw)\cdot\text{MRP}.
\end{equation}
The after-tax SML has a lower intercept ($\rf^A < \rf$) and a flatter
slope ($(1-\tw)\cdot\text{MRP}$), but the same beta. In \textbf{gross
return} space the relationship simplifies to a uniform scaling:
\begin{equation}\label{eq:sml_gross}
  1 + E[R_j^A]
    \;=\; (1-\tw)\bigl(1 + E[R_j]\bigr)
    \qquad \forall\; j.
\end{equation}
The after-tax SML is a vertical contraction of the pre-tax SML by the
factor $(1-\tw)$, mirroring the homothetic contraction of the efficient
frontier in \Cref{fig:contraction_dt}. Both the capital market line (CML)
and the SML contract in the same way: the CML retains the same slope (the
Sharpe ratio is preserved by \eqref{eq:SR_preservation}) but shifts to a
lower intercept $\rf^A$; the SML retains the same betas but compresses both
intercept and slope by $(1-\tw)$. \Cref{fig:sml_capm} illustrates the
geometry.

\begin{figure}[ht]
\centering
\begin{tikzpicture}[>=Stealth, scale=1]
\begin{axis}[
  width=0.85\textwidth, height=0.6\textwidth,
  xlabel={$\beta$},
  ylabel={$E[1+R]$},
  ylabel style={at={(axis description cs:-0.12,0.5)}},
  xmin=0, xmax=2.15, ymin=0.75, ymax=1.35,
  xtick=\empty,
  ytick=\empty,
  axis lines=left,
  clip=false
]

% --- Pre-tax SML: 1+E[R] = 1.04 + 0.06*beta ---
\addplot[blue, thick, domain=0:2, samples=2]
  ({x}, {1.04 + 0.06*x});

% --- After-tax SML: (1-tw)*(1+E[R]) = 0.85*(1.04 + 0.06*x) ---
\addplot[red!70!black, thick, dashed, domain=0:2, samples=2]
  ({x}, {0.85*(1.04 + 0.06*x)});

% --- Fama's implied SML: 1+E[R] = (1.04 + 0.15) + 0.06*beta ---
\addplot[black!50, thick, densely dotted, domain=0:2, samples=2]
  ({x}, {1.19 + 0.06*x});

% --- Risk-free points (on the y-axis) ---
\addplot[mark=*, mark size=2.5pt, blue, only marks]
  coordinates {(0, 1.04)};
\addplot[mark=*, mark size=2.5pt, red!70!black, only marks]
  coordinates {(0, 0.884)};
\addplot[mark=*, mark size=2.5pt, black!50, only marks]
  coordinates {(0, 1.19)};

% --- Market portfolio points (beta=1) ---
\addplot[mark=*, mark size=3pt, blue, only marks]
  coordinates {(1, 1.10)};
\addplot[mark=*, mark size=3pt, red!70!black, only marks]
  coordinates {(1, 0.935)};

% --- Risk-free labels (to the left of y-axis) ---
\node[blue, anchor=east, font=\small] at (axis cs:-0.05, 1.04)
  {$1{+}\rf$};
\node[red!70!black, anchor=east, font=\small] at (axis cs:-0.05, 0.884)
  {$(1{-}\tw)(1{+}\rf)$};
\node[black!50, anchor=east, font=\small] at (axis cs:-0.05, 1.19)
  {$1{+}\rf{+}\tw$};

% --- Market portfolio labels ---
\node[blue, font=\small, anchor=south west] at (axis cs:1.03, 1.10)
  {$M$};
\node[red!70!black, font=\small, anchor=north west] at (axis cs:1.03, 0.935)
  {$M^A$};

% --- Line labels ---
\node[blue, font=\small, anchor=south] at (axis cs:1.75, 1.145)
  {Pre-tax SML};
\node[red!70!black, font=\small, anchor=north] at (axis cs:1.75, 0.975)
  {After-tax SML};
\node[black!50, font=\small, anchor=south] at (axis cs:1.6, 1.29)
  {\citeauthor{Fama2021}};

% --- Arrow: correct contraction (pre-tax -> after-tax) ---
\draw[->, thick, red!70!black, line width=1pt]
  (axis cs:1, 1.08) -- (axis cs:1, 0.955);
\node[red!70!black, font=\footnotesize, anchor=west] at (axis cs:1.05, 1.015)
  {$(1{-}\tw)$};

% --- Arrow: Fama's shift (pre-tax -> Fama) ---
\draw[->, thick, black!50, line width=1pt]
  (axis cs:1, 1.12) -- (axis cs:1, 1.23);
\node[black!50, font=\footnotesize, anchor=west] at (axis cs:1.05, 1.175)
  {${+}\tw$};

\end{axis}
\end{tikzpicture}
\caption{The security market line under a proportional wealth tax. The pre-tax
SML (solid blue) connects the risk-free gross return $1{+}\rf$ to the market
portfolio~$M$. The correct after-tax SML (dashed red) is a uniform vertical
contraction by $(1{-}\tw)$: both the intercept and the slope are scaled down,
preserving the Sharpe ratio and all betas. \citeauthor{Fama2021}'s
(\citeyear{Fama2021}) implicit SML (dotted grey) shifts upward by~$\tw$,
adding the wealth tax to the cost of capital while keeping the risk-free rate
unchanged; this overstates the required return because it does not account for
the wealth tax on the discount rate itself. Parameters:
$\rf = 0.04$, $\text{MRP} = 0.06$, $\tw = 0.15$ (exaggerated for
visibility).}
\label{fig:sml_capm}
\end{figure}
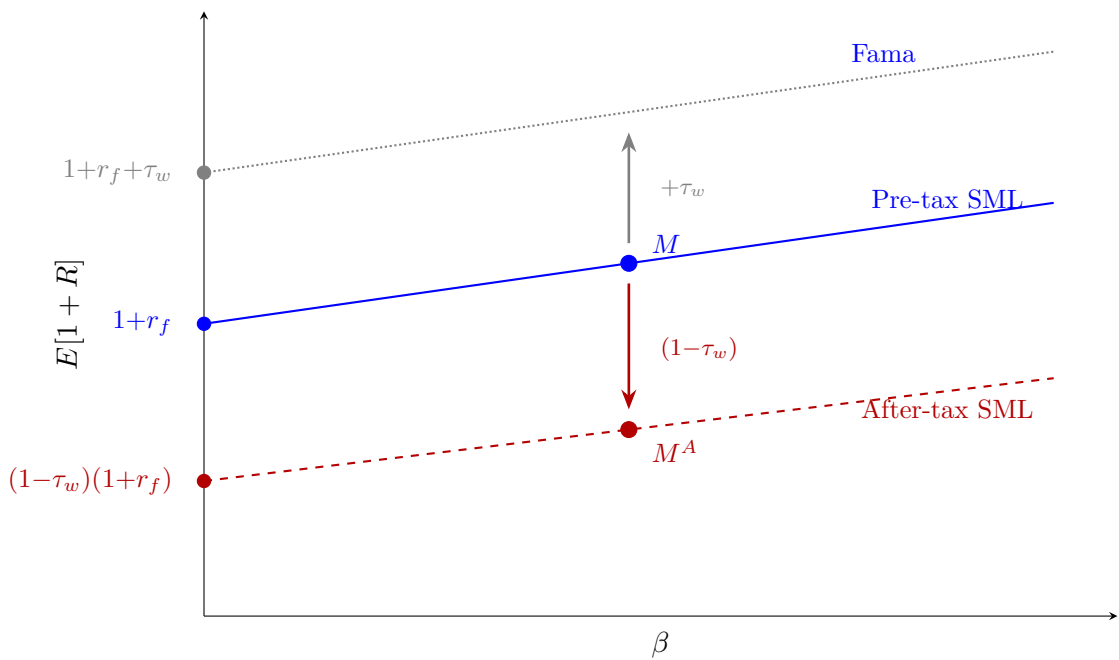

\subsubsection{From partial to general equilibrium}\label{sec:capm_pe_ge}

The results above are partial equilibrium: the return distribution
$(\muvec, \V)$ is taken as given. We now close the model by imposing CAPM
market clearing.

In the standard CAPM derivation, each investor maximises expected utility
subject to a budget constraint, and market clearing requires that aggregate
demand equals the market capitalisation vector. Under CRRA preferences,
\Cref{prop:portfolio_gen}(b) establishes that the optimal weight
vector~$\w^*$ is independent of~$\tw$. Since the portfolio demand functions
are unaffected by the wealth tax, market clearing produces the same
equilibrium expected returns and the same asset prices as in the untaxed
economy:
\begin{equation}\label{eq:capm_ge_invariance}
  \muvec^{\text{GE, tax}} = \muvec^{\text{GE, no tax}}, \qquad
  \V^{\text{GE, tax}} = \V^{\text{GE, no tax}}.
\end{equation}
The partial-equilibrium assumption~(A4)---that $(\muvec, \V)$ does not
depend on the tax rate---is therefore \emph{self-fulfilling} under CRRA.
This sharpens the caveat in \Cref{sec:conditions_economic}: the wealth tax
does not disturb general equilibrium when preferences are CRRA and the tax
is proportional and universal.

The result extends to a mixed economy with both taxed and untaxed investors.
Since both investor types hold the same portfolio weights
(\Cref{prop:portfolio_gen}), aggregate demand is invariant to the fraction
of wealth held by each type. The wealth tax is invisible in CAPM general
equilibrium: it transfers a fraction~$\tw$ of end-of-period wealth from
investors to the government without distorting any price or allocation.

\subsubsection{The error in Fama (2021)}\label{sec:fama_error}

\citet{Fama2021} argues that a wealth tax lowers asset prices. In a
Gordon-growth perpetuity, his mechanism is:
\begin{equation}\label{eq:fama_pricing}
  P_{\textup{Fama}} = \frac{D}{r + \tw} < \frac{D}{r} = P^0.
\end{equation}
The wealth tax adds~$\tw$ to the investor's required pre-tax return,
raising the effective discount rate from $r$ to $r + \tw$. In CAPM terms,
this amounts to shifting individual assets upward along a fixed SML---the
dotted grey line in \Cref{fig:sml_capm}.

The error is that the SML is not fixed. The wealth tax contracts the
\emph{entire} after-tax opportunity set: the risk-free rate falls from
$\rf$ to $\rf^A = \rf - \tw(1+\rf)$, and the market risk premium falls
from $\text{MRP}$ to $(1-\tw)\cdot\text{MRP}$. The investor cannot demand
the old after-tax return because every alternative---including the risk-free
asset---is also taxed. In the notation of \Cref{sec:mm_perspective}, the
correct after-tax discount rate is $k^A = (1-\tw)(1+k) - 1$, which is
\emph{lower} than $k$, not higher. Both the expected return and the
opportunity cost of capital are reduced by the factor $(1-\tw)$, so the
ratio---and hence the price---is preserved ($V = V^0$).

Fama's result is internally consistent only if the investor has access to an
untaxed outside option---a foreign opportunity set that pins the after-tax
required return at the pre-tax level. This is a partial-equilibrium
assumption in which the foreign investor's returns are exogenous and
the domestic investor is a price-taker against an untaxed benchmark. It
produces the same formula as \citet{JohnsenLensberg2014}:
$V = (k-\tw)/k \cdot V^0$. When the wealth tax applies to all domestic
assets including the risk-free rate, the outside option vanishes, both the
numerator and denominator of the valuation are scaled by $(1-\tw)$, and
$V = V^0$.

\begin{remark}[Scope of the Fama critique]
\citeauthor{Fama2021}'s conclusion that wealth taxes lower prices is not
wrong \emph{per se}---it follows from his assumption that asset returns must
clear against an untaxed benchmark. The error lies in presenting this as a
general equilibrium result when it is, in fact, a partial equilibrium
statement: it requires an exogenous opportunity set unaffected by the tax.
When the tax is universal, no such benchmark exists, and the price effect
vanishes.
\end{remark}

\begin{remark}[Residence-based wealth taxes and the ``untaxed outside option'']
\label{rem:residence_based}
A common argument in the policy debate is that, in a small open economy,
the marginal price-setting investor is a foreign (untaxed) investor, and
that the relevant opportunity set is therefore untaxed. Under this view, a
domestic wealth tax raises the required return for domestic investors, who
must compete against an untaxed foreign benchmark---exactly Fama's mechanism.

This argument overlooks a fundamental feature of how wealth taxes are
implemented in practice. All OECD countries that levy (or have levied) a net
wealth tax do so on a \emph{residence basis}: the tax applies to the
investor's \emph{worldwide} assets, not merely to domestic
holdings~\citep{OECD2018}. A Norwegian investor who buys US equities,
German bonds, or any other foreign asset still pays the wealth tax on those
holdings. No asset accessible to a domestic resident is exempt---there is no
``untaxed outside option'' in the investor's opportunity set.

This has three implications. First, the universality condition
(\Cref{sec:conditions_economic}) is satisfied automatically for any taxed
investor under a residence-based wealth tax on market values: the entire
opportunity set---domestic and foreign, risky and risk-free---is taxed at
the same rate. Second, the distinction between partial and general
equilibrium collapses. The taxed investor's portfolio weights are unchanged
(\Cref{prop:portfolio_gen}), the untaxed foreign investor's demand is
trivially unchanged, and aggregate demand is therefore invariant to the
tax---so equilibrium prices do not adjust. For a small open economy, where
domestic demand is negligible in global markets, this holds even without the
CRRA assumption. Third, even emigration---changing tax residency from taxed
to untaxed---does not affect asset prices, because the paper's pricing
neutrality (\Cref{prop:pricing_gen}) establishes that both investor types
are willing to pay the same price per share. The migration channel affects
government revenue but not market prices or portfolio allocations.

\citeauthor{Fama2021}'s price effect $P = D/(r + \tw)$ therefore requires
not merely that some investors are untaxed, but that the \emph{taxed}
investor can access an asset whose return is not subject to the wealth
tax---a condition that no existing or historical OECD wealth tax permits for
its residents. The channels through which a wealth tax \emph{does} affect
prices are the non-universality channels of \Cref{sec:beyond_neutrality}:
book-value taxation (which creates heterogeneous effective rates across
assets), liquidity frictions, and dividend extraction.
\end{remark}

\subsubsection{When universality fails: the Security Market
Fan}\label{sec:capm_smf}

The CAPM framework also illuminates the equilibrium consequences when
universality fails. Under book-value taxation (\Cref{sec:bv_taxation}), the
effective tax rate on asset~$j$ is $\tw\theta_j$ where
$\theta_j = B_j / V_j$ varies across assets. This creates exactly the
heterogeneous-tax setting analysed by \citet{EiksethLindset2009}, who show
that the SML fans out into a family of lines---one for each effective tax
rate---which they term the \emph{Security Market Fan}. In their framework,
the after-tax beta becomes
\begin{equation}\label{eq:el_beta}
  \beta_j^\tau
    \;=\; \frac{(1-\tau_j)\,\beta_j}
               {1 - \textstyle\sum_i w_{i,M}\,\tau_i\,\beta_i}
\end{equation}
where $\tau_j$ is asset~$j$'s effective tax rate and $w_{i,M}$ its weight
in the market portfolio. For a uniform $\tau_j = \tw$ across all assets,
the numerator and denominator scale by the same factor and
$\beta_j^\tau = \beta_j$, recovering \eqref{eq:beta_aftertax}. For
heterogeneous $\tau_j = \tw\theta_j$, assets with low book-to-market ratios
(high~$\theta_j$) see their after-tax betas compressed, while assets with
high book-to-market ratios (low~$\theta_j$) retain betas closer to the
pre-tax values. The SML fans out, and the equilibrium risk-return
relationship is no longer a single line.

This provides the equilibrium mechanism behind the beta-dependent GE
pricing formula in \Cref{prop:bv_multiperiod}(b): book-value taxation
creates heterogeneous effective tax rates across assets, which reshapes the
SML from a line into a fan and introduces the $(1-\beta)$ term in the
valuation formula. \citet{Sandvik2016} makes the same point for Norwegian
unlisted shares: the measured beta increases by $1/\theta$ but the risk
premium per unit of beta falls by~$\theta$, leaving the product for each
individual asset unchanged while distorting the cross-sectional ranking.

\subsection{The Modigliani-Miller Perspective: The Tax Claim as a Separate
Security}\label{sec:mm_perspective}

The pricing neutrality result (\Cref{prop:pricing_gen}) can be derived from a
complementary angle using the Modigliani-Miller framework
\citep{ModiglianiMiller1958}. This perspective treats the government's wealth
tax claim as a separate security backed by the firm's cash flows, analogous to
how debt is treated in the classical MM analysis of capital structure
\citep{Hamada1972,Rubinstein1973}.

\textbf{Setup.} Consider a firm generating a perpetual stream of random cash
flows~$x_t$ with $E[x_t] = \bar{x}$ and systematic risk $\beta_U$. In the
absence of taxes, the firm's value and cost of capital are
\[
V^0 = \frac{\bar{x}}{k}, \qquad k = \rf + \beta_U \mu,
\]
where $\mu$ is the market risk premium. Now introduce a proportional wealth
tax at rate~$\tw$ on the firm's market value. The government receives~$\tw V$
per period, where $V$ is the post-tax market value. The firm's pre-tax cash
flows are now split between two claimants: the equity holder and the
government.

\textbf{The tax claim has the same risk as the firm.} The tax
payment~$\tw V$ is proportional to the firm's market value, so the
government's claim has the same systematic risk as the underlying asset:
$\beta_{\text{tax}} = \beta_U$. Using the \citet{Hamada1972} leverage formula
with the tax claim playing the role of ``debt,''
\[
\beta_{\text{equity}}
= \frac{V^0}{V}\,\beta_U - \frac{T}{V}\,\beta_{\text{tax}}
= \frac{V^0}{V}\,\beta_U - \frac{T}{V}\,\beta_U
= \beta_U.
\]
The equity beta is unchanged. The government becomes a silent proportional
partner, sharing both upside and downside pro rata. Removing a proportional
slice does not change the risk profile of what remains.

\textbf{The discount rate must reflect the tax.} The wealth tax contracts the
entire after-tax opportunity set by the factor~$(1-\tw)$, as established in
\Cref{sec:gen_scalar}. Both the risk-free gross return and the risk premium
are scaled:
\begin{align}
1 + \rf^A &= (1-\tw)(1+\rf), \label{eq:mm_rf_after}\\
\mu^A &= (1-\tw)\,\mu. \label{eq:mm_mu_after}
\end{align}
The after-tax cost of capital for a $\beta_U$-risk asset is therefore
\begin{equation}\label{eq:mm_kA}
k^A = \rf^A + \beta_U\,\mu^A = (1-\tw)(1+k) - 1.
\end{equation}
Crucially, $k^A \neq k$: the wealth tax reduces the investor's discount rate
together with the expected cash flows.

\textbf{Pricing neutrality from the MM perspective.} The after-tax investor
holds a position worth~$V$ at the start of the period. At year-end, the
total position before tax is $(x + V)$ (dividend plus continuation value),
and the tax takes $\tw(x+V)$, leaving $(1-\tw)(x + V)$. In a stationary
equilibrium the investor reinvests~$V$, so the required-return condition is
\[
V(1+k^A) = (1-\tw)\bigl(E[x] + V\bigr).
\]
Rearranging:
\[
V\bigl(k^A + \tw\bigr) = (1-\tw)\,E[x] = (1-\tw)\,k\,V^0.
\]
Now, $k^A + \tw = (1-\tw)(1+k) - 1 + \tw = (1-\tw)\,k$, giving
\begin{equation}\label{eq:mm_pricing}
V = \frac{(1-\tw)\,k\,V^0}{(1-\tw)\,k} = V^0.
\end{equation}
The market price is unchanged. The mechanism is transparent: the
factor~$(1-\tw)$ appears in both the numerator (reduced after-tax cash flows)
and the denominator (reduced opportunity cost of capital), and cancels
exactly. This is the valuation counterpart of the orthogonality result
(\Cref{prop:orthogonality_gen}) and confirms
\Cref{prop:pricing_gen} from the MM side.

\textbf{Comparison with \citet{JohnsenLensberg2014}.} Johnsen and Lensberg
model the wealth tax as a separate claim in an MM framework applied to
Norwegian listed firms. For the case where the tax base is market value,
they derive
\begin{equation}\label{eq:jl_formula}
V_{JL} = V^0\!\left(1 - \frac{\tw}{k}\right).
\end{equation}
This formula subtracts the tax from the cash flows
($E[x] - \tw V^0 = (k-\tw)V^0$) but retains the pre-tax discount
rate~$k$ in the denominator. In their partial equilibrium, where untaxed
foreign investors set both the market price and the discount rate, this
approach is internally motivated. However, the underlying inconsistency is
that the cash flows are reduced by the tax while the discount rate is not:
\[
V_{JL} = \frac{\overbrace{(k-\tw)}^{\text{after-tax CF}}}
              {\underbrace{k\vphantom{()}}_{\text{pre-tax rate}}}
         \; V^0.
\]
Using the after-tax discount rate~$k^A$ from \eqref{eq:mm_kA} for the same
after-tax cash flows yields $V = V^0$, as shown in~\eqref{eq:mm_pricing}.
The discrepancy arises because the wealth tax contracts the \emph{entire}
after-tax opportunity set---not just the cash flow from the asset in
question. Both the expected return and the opportunity cost of capital are
reduced by the same factor~$(1-\tw)$, so the ratio---and hence the
price---is preserved. When the tax base is book value rather than market
value, two additional errors in JL's framework emerge; see
\Cref{sec:bv_jl}.

\begin{remark}[Tax base specification]
The neutrality of the MM decomposition depends on the tax being proportional
to market value, so that $\beta_{\emph{tax}} = \beta_U$. If the tax base
diverges from market value---for example, book value, assessed value, or
lagged historical cost---the tax claim acquires a different beta, and the
\citet{Hamada1972} formula produces an equity beta that differs
from~$\beta_U$. In particular, a tax on book value (with
$\beta_{\text{tax}} \approx 0$) acts like riskless debt, creating a
leverage-like increase in equity beta
\citep[cf.][Section~2.3]{JohnsenLensberg2014}. The tax base specification
discussed in \Cref{sec:conditions_economic} is therefore material for the
MM decomposition as well. \Cref{sec:bv_taxation} refines this observation:
$\beta_{\text{tax}} \approx 0$ is correct for the one-period tax obligation,
but $\beta_{\text{tax}} = \beta_U$ for the perpetual tax claim under a
stationary book-to-market ratio.
\end{remark}

\subsection{CRRA and the Wealth Effect}\label{sec:crra}

For CRRA utility with $\gamma \neq 1$, the value function depends on the
wealth level. Since the wealth tax reduces wealth, it can in principle affect
portfolio choice through the wealth effect on risk aversion, particularly when
the investor optimises consumption jointly with portfolio allocation
\citep[the Merton problem;][]{Merton1969,Merton1971}. However, the CRRA structure ensures that the optimal
\emph{portfolio weights} are invariant to the wealth level, and hence to the
tax. This is the content of \Cref{prop:portfolio_gen}(b). For non-CRRA
preferences (e.g.\ CARA or habit formation), the wealth effect can break
portfolio invariance.

Two supplementary discussions are deferred to appendices: a detailed
comparison of the continuous-time and discrete-time formulations
(\Cref{app:ct_dt}) and a survey of the related literature
(\Cref{app:literature}).

% =========================================================================
\section{Beyond Neutrality: Three Channels of Non-Neutral Taxation}%
\label{sec:beyond_neutrality}
% =========================================================================

\subsection{Overview}\label{sec:beyond_overview}

The neutrality results in
Sections~\ref{sec:single_gbm}--\ref{sec:generalisation} hold under a specific
set of economic conditions, identified in \Cref{sec:conditions_economic}. Each
condition, if relaxed, opens a channel through which the proportional wealth tax
\emph{does} affect asset prices, portfolio choice, or both. We now formalise
three such channels that are of particular practical importance: (i)~book-value
taxation, (ii)~liquidity frictions, and (iii)~dividend extraction and
investment distortion.

These three channels share a common structural mechanism: each violates the
\emph{universality} condition---the requirement that all assets bear the same
effective tax rate on market value. When this condition holds, the wealth tax
operates as a multiplicative scalar $(1-\tw)^n$ that cancels in every ratio,
preserving relative returns, Sharpe ratios, and portfolio weights
(\Cref{sec:gen_scalar}). When universality fails, the effective tax rate
varies across assets, the multiplicative structure breaks, and the tax
distorts relative prices and allocations.

\Cref{tab:channels} summarises the three channels, the condition each relaxes,
and the direction of the pricing effect relative to the no-tax value~$V^0$.

\begin{table}[ht]
\centering
\caption{Three channels of non-neutral wealth taxation.}
\label{tab:channels}
\small
\begin{tabular}{@{}llll@{}}
\toprule
Channel & Condition relaxed & Direction & Key mechanism \\
\midrule
Book value (\Cref{sec:bv_taxation})
  & Universality (tax base)
  & $V > V^0$ when $\theta < 1$
  & Effective rate $\tw\theta_j$ varies \\
Liquidity (\Cref{sec:liq_frictions})
  & Frictionless markets
  & $V < V^0$
  & Forced selling at cost $c_j$ \\
Dividends (\Cref{sec:div_extraction})
  & Frictionless markets
  & $V < V^0$
  & Foregone investment to fund tax \\
\bottomrule
\end{tabular}
\end{table}

The book-value channel is fully formalised using the no-arbitrage pricing
framework of \Cref{sec:mm_perspective}. It yields closed-form results for both
one-period and multi-period settings, and identifies specific errors in the
existing Norwegian literature on wealth tax effects. The liquidity and dividend
channels are developed within the same framework, though with more stylised
assumptions reflecting the current state of formal knowledge.

% -------------------------------------------------------------------------
\subsection{Book-Value Taxation}\label{sec:bv_taxation}
% -------------------------------------------------------------------------

\subsubsection{Motivation and institutional background}%
\label{sec:bv_motivation}

In many countries that impose a wealth tax, the tax base deviates from market
value for certain asset classes. In Norway, listed shares are taxed at market
value, but unlisted shares are taxed at their book value (the firm's equity
per share as reported in the tax return), and real estate is taxed at an
assessed value that is typically below market value.%
\footnote{The Norwegian wealth tax currently uses approximately 80\% of
assessed value for primary residences and 100\% of assessed value for
secondary housing, with assessed values lagging market values. Unlisted
shares are valued at their book equity per the company's latest tax balance.
Listed shares have been valued at market since 2006.}

Define the \emph{book-to-market ratio}
\begin{equation}\label{eq:bv_theta}
  \theta_j \equiv \frac{B_j}{V_j},
\end{equation}
where $B_j$ is the book (or assessed) value per share and $V_j$ is the market
value. When the wealth tax $\tw$ is levied on book value rather than market
value, the effective tax rate on market value is $\tw\theta_j$, which varies
across assets. Universality is violated: the tax creates differential costs
across assets in exactly the way flagged in \Cref{sec:conditions_economic}.

This section formalises the pricing effect of book-value taxation within the
no-arbitrage framework of \Cref{sec:mm_perspective}. As established in
\Cref{sec:independence_capm}, the results follow from no-arbitrage alone;
where a general equilibrium extension requires additional structure (CAPM),
this is stated explicitly.

The starting point is the Remark on tax base specification in
\Cref{sec:mm_perspective}, which observed that a tax on book value (with
$\beta_{\text{tax}} \approx 0$) acts like riskless debt in the
\citet{Hamada1972} framework. We now show that this observation is correct for
the one-period tax obligation but requires refinement for the perpetual tax
claim.

\subsubsection{Three observations}\label{sec:bv_observations}

The formalisation rests on three observations about the nature of book-value
tax claims.

\medskip
\noindent\textbf{Observation~1} (Lagged book value is deterministic).
Under the Norwegian tax code, the wealth tax liability for year~$t$ is based
on the book value at the \emph{end of year~$t-1$}. At the time the investor
makes portfolio decisions for period~$t$, the tax base~$B_{t-1}$ is already
known. The one-period tax obligation $\tw B_{t-1}$ is therefore a
deterministic cash flow from the investor's perspective.

\medskip
\noindent\textbf{Observation~2} (Empirical stationarity of~$\theta$).
The book-to-market ratio~$\theta$ is empirically highly persistent.
\citet{PontiffSchall1998} report a monthly first-order autocorrelation
of~0.97 for the aggregate book-to-market ratio, and find that book values
and market values are cointegrated through future cash flows. The residual
income model of \citet{Ohlson1995} provides the theoretical underpinning:
the clean surplus relation links book value to market value, so $\theta$
reverts to a long-run mean determined by the firm's return on equity
relative to its cost of capital.

\medskip
\noindent\textbf{Observation~3} (Gordon steady state as special case).
If the firm is in a Gordon steady state with all earnings paid out
(no growth), then $\theta$ is exactly constant:
$\theta = k / \text{ROE}$, where ROE is the return on equity
and $k$ is the cost of equity. In this case, the
perpetual tax claim $\tw B_t = \tw \theta V_t$ is proportional to market
value with a \emph{constant} ratio~$\theta$, and the beta of the tax claim
equals the beta of the underlying asset: $\beta_{\text{tax}} = \beta_U$.
This contrasts with the one-period case where
$\beta_{\text{tax}} \approx 0$.

\subsubsection{One-period result}\label{sec:bv_oneperiod}

We begin with the one-period case, where the result is model-free.

Consider an asset with random end-of-period payoff~$x$ (inclusive of
dividends and terminal value), current market value~$V$, and no-tax
value~$V^0$. The one-period risk-free rate is~$\rf$. The investor holds the
asset for one period and pays wealth tax $\tw B$ on the known book
value~$B$ at period end.%
\footnote{We write $B$ without a time subscript because it is known at the
start of the period; it is the lagged book value from the previous year-end.}

\begin{theorem}[One-Period Book-Value Pricing]\label{thm:bv_oneperiod}
Under no-arbitrage, the market value of the asset under book-value wealth
taxation is
\begin{equation}\label{eq:bv_oneperiod}
  V = \frac{1}{1-\tw}\left(V^0 - \frac{\tw B}{1+\rf}\right),
\end{equation}
provided $\theta < 1 + \rf$.
\end{theorem}

\begin{proof}[Proof sketch]
The after-tax cash flow is $x - \tw B$, where $\tw B$ is deterministic by
Observation~1. Under the risk-neutral measure~$\mathbb{Q}$ (which is
unaffected by the introduction of a deterministic liability), the
no-arbitrage price satisfies
\[
  V = \frac{1}{1 + \rf^A}
      \left(E^{\mathbb{Q}}[x] - \tw B\right),
\]
where $\rf^A = \rf - \tw(1+\rf)$ is the after-tax risk-free rate from the
main text. Since $E^{\mathbb{Q}}[x] = V^0(1+\rf)$ by definition of the
no-tax value, substitution and simplification yield~\eqref{eq:bv_oneperiod}.
The condition $\theta < 1 + \rf$ ensures $V > 0$; it is satisfied whenever
the book value is less than the future value of the market price, which
holds empirically for virtually all assets.
\end{proof}

\begin{remark}[Generality]
Like \Cref{prop:pricing_gen}, this result is distribution-free and holds
under any asset pricing model consistent with no-arbitrage. The only
additional requirement is that the tax base~$B$ is deterministic---which, by
Observation~1, is an institutional fact, not a modelling assumption.
\end{remark}

\begin{remark}[Direction of the effect]
When $\theta < 1$ (book value below market value, the empirically dominant
case), the taxed asset is worth \emph{more} than under market-value
taxation: the investor pays tax on a smaller base, so the effective tax
burden is lighter. In the limit $\theta \to 0$ (negligible book value), the
tax vanishes entirely and $V \to V^0 / (1-\tw)$---the full benefit of a
tax-exempt asset. When $\theta = 1$ (book equals market), the formula
reduces to $V = V^0$, recovering market-value neutrality.
\end{remark}

\subsubsection{Multi-period extension}\label{sec:bv_multiperiod}

The one-period result extends to a perpetual setting by combining
Observations~2 and~3 with the Gordon steady-state model. We consider a firm that
pays out all earnings in perpetuity, with cost of
equity~$k$ and a stationary book-to-market ratio~$\theta$. We develop two
cases: partial equilibrium (PE), where the discount rate is unaffected by
the tax, and general equilibrium (GE), where the market risk premium adjusts
because the market portfolio itself is subject to book-value taxation.

\begin{proposition}[Multi-Period Book-Value Pricing]\label{prop:bv_multiperiod}
Under no-arbitrage and a stationary book-to-market ratio~$\theta$:
\begin{enumerate}
  \item[\textup{(a)}] \textbf{Partial equilibrium.} If the discount rate~$k$
  is unaffected by the tax,
  \begin{equation}\label{eq:bv_perpetuity_pe}
    V_{\textup{PE}} = \frac{k}{k - \tw(1+\rf-\theta)}\; V^0.
  \end{equation}

  \item[\textup{(b)}] \textbf{General equilibrium.} If the market risk
  premium adjusts so that the after-tax required return reflects the
  systematic risk of the tax claim, and the asset has market beta~$\beta$,
  \begin{equation}\label{eq:bv_perpetuity_ge}
    V_{\textup{GE}} = \frac{k}
      {k - \tw(1-\beta)(1+\rf-\theta)}\; V^0.
  \end{equation}
\end{enumerate}
\end{proposition}

\begin{proof}[Derivation sketch]
The required return under book-value taxation satisfies
$k^A = k - \tw(1+\rf) + \tw\theta$, where the first two terms are the
standard after-tax discount rate from the main text and the last term
reflects the \emph{lower} effective tax rate due to the book-value base.
The PE pricing follows from $V = E[x]/k^A$ with $E[x] = kV^0$.

For the GE case, the market portfolio is itself subject to book-value
taxation. Decomposing the required return into risk-free and risk-premium
components, and noting that the risk-free rate adjusts fully (the risk-free
asset has $\theta = 1$ if taxed at market value) while the risk premium
adjusts only for the $(1-\beta)$ component that reflects the book-value
discount, yields~\eqref{eq:bv_perpetuity_ge}. This decomposition uses the
CAPM structure to separate the discount rate into its risk-free and
systematic risk components---the only point at which an equilibrium pricing
model enters.
\end{proof}

\begin{remark}[PE vs.\ GE]
The PE formula~\eqref{eq:bv_perpetuity_pe} is beta-independent: it uses only
the cost of equity~$k$ and the book-to-market ratio~$\theta$, and holds
under any pricing model. The GE formula~\eqref{eq:bv_perpetuity_ge}
introduces beta-dependence because the market risk premium itself adjusts.
For $\beta = 1$ (the market portfolio), the GE formula gives $V = V^0$: the
book-value benefit is fully absorbed into a lower equilibrium required return,
restoring neutrality for the market as a whole. For $\beta < 1$, the asset
retains a net benefit ($V > V^0$); for $\beta > 1$, the required return
falls by more than the tax benefit, and $V < V^0$. The underlying
equilibrium mechanism is the Security Market Fan of
\citet{EiksethLindset2009}: heterogeneous effective tax rates
$\tw\theta_j$ across assets reshape the SML from a single line into a fan,
producing the beta-dependent pricing; see \Cref{sec:capm_smf}.
\end{remark}

\begin{remark}[Connection to the tax base Remark]
The Remark on tax base specification in \Cref{sec:mm_perspective} observed
that a book-value tax claim with $\beta_{\text{tax}} \approx 0$ acts like
riskless debt. This is precisely the one-period result
(\Cref{thm:bv_oneperiod}), where the deterministic tax $\tw B$ is discounted
at the risk-free rate. The multi-period result refines this: when~$\theta$ is
stationary, the perpetual tax claim has $\beta_{\text{tax}} = \beta_U$
(Observation~3), and the leverage-like effect identified in the Remark
disappears for the capitalised claim. The transition from
$\beta_{\text{tax}} \approx 0$ (one-period) to $\beta_{\text{tax}} = \beta_U$
(perpetuity) is the key structural difference.
\end{remark}

\subsubsection{Comparison with Johnsen and Lensberg}\label{sec:bv_jl}

\citet{JohnsenLensberg2014} analyse the effect of book-value wealth taxation
on equity prices within a partial equilibrium CAPM framework. Their central
result is a pricing formula
\[
  V_{JL} = \frac{k - \tw}{k}\; V^0,
\]
which they use to argue that the wealth tax reduces asset values by
approximately 10\% (for Norwegian tax rates). This formula, already discussed
in \Cref{sec:mm_perspective}, contains two compounding errors.

\begin{corollary}[Identification of errors in JL]\label{cor:jl_errors}
The \citet{JohnsenLensberg2014} pricing formula contains two errors relative
to the no-arbitrage pricing derived in \Cref{thm:bv_oneperiod} and
\Cref{prop:bv_multiperiod}:
\begin{enumerate}
  \item[\textup{(i)}] \textbf{Pre-tax discount rate.} JL discount
  after-tax cash flows at the pre-tax rate~$k$, whereas no-arbitrage
  requires the after-tax rate $k^A = k - \tw(1+\rf) + \tw\theta$. This
  double-counts the tax: cash flows are reduced by the tax, but the
  discount rate does not reflect the corresponding reduction in opportunity
  cost.

  \item[\textup{(ii)}] \textbf{Risk of the tax claim.} JL implicitly treat
  the tax claim as having $\beta_{\textup{tax}} = 0$ (the one-period
  characterisation), but apply this to a perpetual setting. Under the
  Gordon steady state with stationary~$\theta$, the correct beta is
  $\beta_{\textup{tax}} = \beta_U$ (Observation~3), so the tax claim carries
  the same systematic risk as the underlying asset. A related critique
  of the JL framework appears in \citet{Sandvik2016}; the one-period
  beta preservation result of \citet{HansenSandvik2022} highlights the
  time-horizon dependence.
\end{enumerate}
These errors compound: (i) overstates the tax burden, and (ii) misallocates
the risk, producing the spurious prediction that the wealth tax reduces
asset values by approximately 10\%.
\end{corollary}

\subsubsection{Numerical illustration}\label{sec:bv_numerical}

\Cref{tab:bv_numerical} illustrates the pricing effect for a representative
parameterisation: $\rf = 0.03$, $k = 0.10$, $\tw = 0.01$,
$\theta = 0.50$, and $V^0 = 100$.

\begin{table}[ht]
\centering
\caption{Book-value wealth tax: pricing effects under alternative models.}
\label{tab:bv_numerical}
\small
\begin{tabular}{@{}lrl@{}}
\toprule
Model & $V$ & Comment \\
\midrule
No tax                         & 100.0 & Benchmark $V^0$ \\
Market-value tax (neutral)     & 100.0 & \Cref{prop:pricing_gen} \\
JL formula                     &  90.0 & $V_{JL} = (k-\tw)/k \cdot V^0$ \\
One-period (\Cref{thm:bv_oneperiod})
  & 100.5 & $(V^0 - \tw B/(1+\rf))/(1-\tw)$ \\
PE perpetuity (\Cref{prop:bv_multiperiod}a)
  & 105.6 & $kV^0/(k - \tw(1+\rf-\theta))$ \\
GE, $\beta = 0.5$             & 102.7 & \eqref{eq:bv_perpetuity_ge} \\
GE, $\beta = 1.0$             & 100.0 & Reduces to $V^0$ \\
GE, $\beta = 1.5$             &  97.4 & \eqref{eq:bv_perpetuity_ge} \\
\bottomrule
\end{tabular}
\end{table}

The sign reversal relative to JL is the most striking feature: JL predict a
10\% value \emph{decrease}, while the correct no-arbitrage pricing in partial
equilibrium yields a 0.5\% to 5.6\% value \emph{increase}. The direction
reverses because the book-value base ($\theta = 0.50$) creates a lighter
effective tax burden than market-value taxation, so assets taxed on book
value are worth more, not less, than under the neutral benchmark.

The magnitude depends on the horizon: the one-period effect is small (0.5\%)
because only the current-year tax obligation benefits from the book-value
discount. The perpetuity effect is larger (5.6\% in PE) because the entire
future stream of tax obligations is discounted, and each one benefits from
$\theta < 1$. The GE beta-dependence arises because the market risk premium
adjusts: for the market portfolio ($\beta = 1$), the book-value benefit is
fully absorbed into a lower required return, so $V = V^0$. For $\beta < 1$
the asset retains a net benefit ($V > V^0$), while for $\beta > 1$ the
required return falls by more than the tax benefit, and $V < V^0$.

% -------------------------------------------------------------------------
\subsection{Liquidity Frictions}\label{sec:liq_frictions}
% -------------------------------------------------------------------------

\subsubsection{Setup}\label{sec:liq_setup}

The main text assumes that the investor can liquidate shares to pay the
wealth tax without cost (\Cref{sec:conditions_economic}, ``Frictionless
rebalancing''). In practice, selling assets to pay a wealth tax incurs
transaction costs that vary across asset classes: listed equities face
bid-ask spreads and market impact costs; real estate faces agent fees,
transfer taxes, and search frictions; private equity positions may be
effectively non-tradeable for extended periods.

We model this by introducing a stochastic illiquidity cost~$c_j$ for
asset~$j$: when the investor sells a fraction of asset~$j$ to pay the wealth
tax, a fraction~$c_j$ of the proceeds is lost to transaction costs. The
effective tax rate on market value becomes $\tw / (1-c_j)$, because the
investor must sell more shares to generate the required after-cost proceeds.
Since~$c_j$ varies across assets and is stochastic, universality is
violated: the effective tax burden differs across assets, and the
multiplicative structure that drives neutrality breaks down.

The liquidity channel differs from the book-value channel in an important
respect: it operates even when the tax base is market value. Book-value
taxation creates non-neutrality through the \emph{definition} of the tax
base; liquidity frictions create non-neutrality through the \emph{mechanics}
of tax payment. Both channels violate universality, but through different
mechanisms.

\begin{remark}[Practical payment mechanics]\label{rem:payment_practice}
The proportional-dilution mechanism is a modelling abstraction that
represents an upper bound on the liquidity effect. In practice, the wealth
tax is rarely paid by selling shares in a single year-end transaction. In Norway,
the personal wealth tax (assessed on 31~December values) is collected
through quarterly advance payments (\emph{forskuddsskatt}) during the
assessment year, smoothing the liquidity demand over time. More importantly,
investors typically cover the tax through a hierarchy of payment sources:
(i)~other liquid assets or income (salary, interest, rental income),
(ii)~dividend payments from the company, and (iii)~share sales---with
outright forced liquidation as a last resort rather than the default
mechanism. Empirical evidence supports this hierarchy:
\citet{BerzinsBohrenStacescu2022} find that Norwegian private firms
increase dividend payouts in response to their owners' wealth tax
obligations, while \citet{Ring2020} documents precautionary saving behaviour
consistent with households maintaining liquid buffers to meet future tax
liabilities. \citet{AlstadsaeterBjornebyKopczukMarkussenRoed2022} observe
that countries implementing wealth taxes make ``practical compromises
regarding treatment and valuation of different asset categories to ease
assessment and liquidity difficulties.'' The forced-selling model should
therefore be interpreted as characterising the marginal cost of the
\emph{least liquid} assets in the portfolio---those for which the payment
hierarchy is exhausted---rather than the average cost across all holdings.
The dividend extraction channel (\Cref{sec:div_extraction}) formalises the
more common payment mechanism.
\end{remark}

\subsubsection{Multi-factor neutrality without friction}%
\label{sec:liq_neutral}

Before introducing frictions, we note that the neutrality results of
Sections~\ref{sec:single_gbm}--\ref{sec:generalisation} extend immediately
to any multi-factor pricing model. Whether expected returns are determined by
a single market factor (CAPM), three factors (Fama--French), or an arbitrary
$K$-factor model with factors $f_1, \ldots, f_K$ (which may include a
liquidity factor), the wealth tax operates as a multiplicative scalar that
reduces all factor premia uniformly:
\[
  E[R_j^A] - \rf^A
  = (1-\tw)\bigl(E[R_j^B] - \rf\bigr)
  = (1-\tw)\sum_{k=1}^K \beta_{j,k}\,\lambda_k,
\]
where $\lambda_k$ is the premium for factor~$k$. The factor betas
$\beta_{j,k}$ are unchanged, and the relative premia across assets are
preserved. Adding liquidity as a priced factor---as in the models of
\citet{AcharyaPedersen2005}, \citet{PastorStambaugh2003}, or
\citet{Amihud2002}---changes nothing about neutrality, as long as the tax
is paid without friction. The key insight is that neutrality is a property
of the \emph{tax mechanism} (multiplicative scaling), not of the asset
pricing model.

\subsubsection{Pricing with friction}\label{sec:liq_pricing}

Introduce a stochastic illiquidity cost~$c_j \in [0,1)$ that the investor
incurs when liquidating shares of asset~$j$ to pay the wealth tax. The
cost~$c_j$ may be correlated with the asset's payoff~$x_j$ and with
aggregate market conditions.

The after-tax cash flow from one share of asset~$j$ is
\[
  x_j - \tw V \cdot \frac{1}{1-c_j},
\]
where the factor $1/(1-c_j)$ reflects the additional shares that must be
sold to cover the transaction cost. Since $c_j$ is stochastic, the after-tax
factor $(1 - \tw/(1-c_j))$ is random, and the clean multiplicative structure
of the main text is lost.

\begin{proposition}[Pricing with Liquidity Frictions]\label{prop:liq_pricing}
Under no-arbitrage, the one-period price of asset~$j$ under wealth taxation
with stochastic illiquidity cost~$c_j$ satisfies
\begin{equation}\label{eq:liq_pricing}
  V \approx V^0
    - \frac{\tw}{1+\rf}\,V \cdot E^{\mathbb{Q}}[c_j]
    - \frac{\tw}{(1+\rf)^2}\,
      \Cov^{\mathbb{Q}}\!\bigl(x_j + V,\; c_j\bigr),
\end{equation}
where $E^{\mathbb{Q}}$ and $\Cov^{\mathbb{Q}}$ denote expectation and
covariance under the risk-neutral measure.
\end{proposition}

\begin{proof}[Derivation sketch]
The no-arbitrage price satisfies
$V = E^{\mathbb{Q}}[x_j - \tw V/(1-c_j)] / (1+\rf^A)$. Expanding
$1/(1-c_j) \approx 1 + c_j$ for small~$c_j$ and separating the
deterministic and stochastic components yields three terms: (i)~the no-tax
value~$V^0$ (from the main text neutrality result), (ii)~a level
effect proportional to~$E^{\mathbb{Q}}[c_j]$ representing the expected
illiquidity cost, and (iii)~a covariance effect capturing the risk premium
from the correlation between asset payoffs and illiquidity costs.
\end{proof}

\begin{remark}[Two effects]
The pricing impact decomposes into a \emph{level effect} and a
\emph{covariance effect}:
\begin{itemize}
  \item \textbf{Level effect} ($E^{\mathbb{Q}}[c_j]$): The expected
  illiquidity cost reduces the asset's value in proportion to the tax rate.
  This is a first-order effect that depends on the asset's average
  liquidity.

  \item \textbf{Covariance effect} ($\Cov^{\mathbb{Q}}(x_j + V, c_j)$):
  If the asset becomes more illiquid precisely when its payoff is low
  (positive covariance between~$c_j$ and losses), the tax-induced forced
  selling is most costly in bad states of the world. This creates a
  liquidity risk premium that is specific to the wealth tax mechanism.
\end{itemize}
\end{remark}

\begin{remark}[Structural insight]
The wealth tax converts a multiplicative (neutral) tax mechanism into a
multiplicative-plus-additive mechanism. In the frictionless case, the
after-tax factor $(1-\tw)$ is a constant multiplier that cancels in ratios.
With friction, the effective factor $(1 - \tw/(1-c_j))$ is stochastic, and
its covariance with asset payoffs introduces an additive,
asset-specific distortion that does not cancel.
\end{remark}

\subsubsection{Connection to the liquidity pricing literature}%
\label{sec:liq_ap}

The covariance effect in \eqref{eq:liq_pricing} maps directly onto the
three liquidity betas identified in the liquidity-adjusted CAPM of
\citet{AcharyaPedersen2005}. In their framework, expected returns reflect
not only the covariance of returns with the market, but also three
additional terms:
\begin{align}
  \beta_{c,c} &: \quad \Cov(c_j, c_M), \label{eq:ap_beta_cc} \\
  \beta_{c,r} &: \quad \Cov(c_j, R_M), \label{eq:ap_beta_cr} \\
  \beta_{r,c} &: \quad \Cov(R_j, c_M), \label{eq:ap_beta_rc}
\end{align}
where $c_M$ denotes aggregate market illiquidity and $R_M$ is the market
return. The first captures commonality in liquidity, the second captures
the tendency for assets to become illiquid when the market rises (or falls),
and the third captures the tendency for returns to be low when the market
is illiquid.

The wealth tax \emph{amplifies} all three channels. Forced selling to pay
the tax creates correlated liquidity demand across all taxed investors,
increasing commonality in illiquidity costs ($\beta_{c,c}$). If
tax-induced selling is concentrated in downturns (when portfolio values
are lower relative to tax liabilities), it increases the covariance between
illiquidity and market conditions ($\beta_{c,r}$ and~$\beta_{r,c}$). The
magnitude depends on the fraction of aggregate ownership subject to the
wealth tax and on the concentration of tax-motivated selling in time.

\begin{remark}
The liquidity channel operates even under market-value taxation. Unlike the
book-value channel, which requires a divergence between the tax base and
market value, the liquidity channel arises from the mechanics of tax
\emph{payment}: any tax that must be paid by selling assets incurs
transaction costs. The effect is therefore a general feature of wealth
taxation, not specific to any particular tax base specification.
\end{remark}

\citet{PastorStambaugh2003} estimate a tradeable liquidity factor with a
risk premium of approximately 7.5\% per year. \citet{Amihud2002} documents a
significant cross-sectional relationship between illiquidity (measured as
price impact per unit of trading volume) and expected returns. The wealth tax
mechanism described here provides a specific channel through which these
liquidity premia are amplified: the tax creates a periodic, predictable
source of forced selling that is correlated across investors and
concentrated in time.

% -------------------------------------------------------------------------
\subsection{Dividend Extraction and Investment Distortion}%
\label{sec:div_extraction}
% -------------------------------------------------------------------------

\subsubsection{Relation to Appendix~A}\label{sec:div_relation}

\Cref{app:dividends} analyses the financial mechanics of dividend
payment vs.\ share sales and shows that the \emph{method} of tax
payment is irrelevant for pricing neutrality (\Cref{app:div_pricing}). The
present section addresses a different question: when the wealth tax forces
the firm (or investor) to extract dividends that would otherwise have been
reinvested, the resulting loss of investment opportunities has a real cost
that breaks pricing neutrality. The non-neutrality arises not from the
financial mechanics of dividend payment, but from the \emph{foregone
investment} that the dividend extraction necessitates.

\subsubsection{Tax payment constraint}\label{sec:div_constraint}

Consider a firm with cost of equity~$k$, payout ratio~$\delta$, and
expected cash flow~$\bar{x} = k V^0$. Under book-value taxation, the
wealth tax liability per share is $\tw\theta V$ per period. If the
investor cannot (or prefers not to) sell shares to pay the tax, the
dividend must cover the tax liability, imposing a minimum payout
constraint:
\begin{equation}\label{eq:div_constraint}
  \delta \geq \frac{\tw\theta}{k}.
\end{equation}
This constraint binds when the firm would otherwise choose a lower payout
ratio---typically growth firms with high reinvestment needs---and when the
investor faces high costs of selling shares, as is the case for private
firms with illiquid equity.

For a concrete illustration: with $\tw = 0.01$, $\theta = 0.50$, and
$k = 0.10$, the minimum payout ratio is $\delta \geq 0.05$, or 5\% of
expected cash flows. For a growth firm that would optimally pay out only
2--3\% (retaining the rest for reinvestment), the wealth tax forces an
increase in the payout ratio that reduces available investment capital.

\subsubsection{Investment distortion}\label{sec:div_investment}

When the payout constraint~\eqref{eq:div_constraint} binds, the firm
foregoes investment opportunities whose internal rate of return exceeds the
cost of capital. The value loss depends on the profitability of the
foregone investment.

\begin{proposition}[Investment Distortion]\label{prop:div_distortion}
If the firm has investment opportunities with internal rate of
return~$\rho > k$, and the payout constraint~\eqref{eq:div_constraint}
binds, the value loss relative to no tax is
\begin{equation}\label{eq:div_distortion}
  \Delta V = \frac{\rho - k}{k}
    \left(\tw\theta\, V - \delta\,\bar{x}\right),
\end{equation}
where $\delta\,\bar{x}$ is the dividend the firm would pay absent the tax
constraint, and $\tw\theta\, V - \delta\,\bar{x}$ is the
additional payout forced by the tax.
\end{proposition}

\begin{proof}[Derivation sketch]
The forced additional payout $\tw\theta\, V - \delta\,\bar{x}$ is
redirected from investment (earning~$\rho$) to the tax authority. Each unit
of foregone investment destroys present value of $(\rho - k)/k$ in
perpetuity. The total value loss is the product of the unit cost and the
volume of displaced investment.
\end{proof}

\begin{remark}
The distortion is larger when (i)~the firm has high growth opportunities
($\rho \gg k$), (ii)~the book-to-market ratio~$\theta$ is high (increasing
the tax burden on book value), and (iii)~the firm's optimal payout
ratio~$\delta$ is low (so the constraint binds more tightly). The
distortion vanishes for firms whose optimal payout already exceeds the tax
threshold---typically mature firms with limited growth opportunities.
\end{remark}

\subsubsection{Optimal payment mechanism}\label{sec:div_optimal}

The investor faces a trade-off between two mechanisms for paying the
wealth tax: selling shares at illiquidity cost~$c_j$, or extracting
dividends at investment distortion cost $(\rho - k)/k$. The optimal
mechanism minimises the total cost:
\begin{equation}\label{eq:div_optimal}
  \text{Cost}_{\text{sell}} = c_j \cdot \tw\theta\, V, \qquad
  \text{Cost}_{\text{dividend}} = \frac{\rho - k}{k}
    \cdot \tw\theta\, V.
\end{equation}
Selling is preferred when $c_j < (\rho-k)/k$ (low illiquidity cost
relative to growth opportunities), and dividend extraction is preferred
when $c_j > (\rho-k)/k$ (high illiquidity cost, limited growth).

This creates a natural partition across asset classes:
\begin{itemize}
  \item \textbf{Listed, mature firms:} Low~$c_j$, low~$\rho$. Sell shares
  to pay tax; minimal distortion.

  \item \textbf{Listed, growth firms:} Low~$c_j$, high~$\rho$. Sell shares;
  moderate distortion from market impact.

  \item \textbf{Private, mature firms:} High~$c_j$, low~$\rho$. Extract
  dividends; moderate distortion.

  \item \textbf{Private, growth firms:} High~$c_j$, high~$\rho$. Both
  mechanisms are costly; maximal distortion. This is the asset class most
  affected by the wealth tax.
\end{itemize}

\citet{BerzinsBohrenStacescu2022} provide direct evidence for this
partition using Norwegian firm-level data. They find that the wealth tax
leads to significantly higher dividend payouts in private firms, with the
effect concentrated among firms with controlling shareholders facing binding
tax obligations. The associated reduction in investment is economically
large: firms with higher wealth tax exposure invest less and grow more
slowly, consistent with the dividend extraction channel formalised here.

% -------------------------------------------------------------------------
\subsection{Combined Effects and Implications}\label{sec:combined_effects}
% -------------------------------------------------------------------------

\subsubsection{Interaction of the three channels}\label{sec:combined_interaction}

The three channels identified above have opposing effects on asset values.
Book-value taxation raises values ($V > V^0$) when the tax base is below
market value ($\theta < 1$), because the effective tax burden is lighter.
Liquidity frictions and dividend extraction both lower values ($V < V^0$),
because they impose real costs on the tax payment process.

The net effect is asset-class dependent. \Cref{tab:combined} summarises the
qualitative direction for the principal asset classes affected by the
Norwegian wealth tax.

\begin{table}[ht]
\centering
\caption{Net pricing effect across asset classes (qualitative).}
\label{tab:combined}
\small
\begin{tabular}{@{}llllll@{}}
\toprule
Asset class & $\theta$ & $c_j$ & Book value & Liquidity & Dividend \\
\midrule
Listed, mature       & 0.5--1.0 & Low  & $+$small & $-$small & $-$negligible \\
Listed, growth       & 0.2--0.5 & Low  & $+$moderate & $-$small & $-$small \\
Private, mature      & 0.5--1.0 & High & $+$small & $-$moderate & $-$moderate \\
Private, growth      & 0.2--0.5 & High & $+$moderate & $-$large & $-$large \\
Real estate          & 0.3--0.8 & High & $+$moderate & $-$moderate & N/A \\
\bottomrule
\end{tabular}
\end{table}

For listed assets with low illiquidity costs, the book-value discount
($\theta < 1$) dominates, and the net effect is likely a modest value
\emph{increase}. For private firms with high illiquidity costs and
significant growth opportunities, the liquidity and dividend channels
dominate, and the net effect is a value \emph{decrease}---potentially
substantial.

\subsubsection{Empirical implications}\label{sec:combined_empirical}

The framework developed in this section explains several empirical
observations that are difficult to reconcile with either pure neutrality or
the predictions of \citet{JohnsenLensberg2014}:

\emph{Real effects concentrated in private firms.}
\citet{BerzinsBohrenStacescu2022} document that the wealth tax leads to
higher dividends and lower investment in private firms, with negligible
effects on listed firms. This is consistent with the asset-class partition
above: listed firms have low illiquidity costs and can sell shares to pay
the tax, while private firms face high~$c_j$ and must extract dividends.

\emph{Book-value discount as partial offset.}
The Norwegian practice of valuing unlisted shares at book value---often 50\%
or less of market value---mitigates the liquidity and dividend costs by
reducing the effective tax rate. The book-value channel provides a
\emph{partial offset} to the real costs imposed by the other two channels.
This explains why the observed effects, while significant, are smaller than
a na\"ive calculation based on the statutory wealth tax rate would suggest.

\emph{Precautionary saving.}
\citet{Ring2020} finds evidence that the wealth tax induces precautionary
saving, with households maintaining higher liquid asset balances to ensure
they can meet future tax obligations. This is consistent with the liquidity
channel: the expected cost of forced selling creates an incentive to hold
more liquid assets, even at the expense of lower expected returns.

\subsubsection{Policy implications}\label{sec:combined_policy}

The analysis yields several insights for tax policy design. First, the
welfare cost of a wealth tax is not uniform across asset classes. The
distortionary effects are concentrated among illiquid assets and growth
firms, while liquid, mature assets are largely unaffected.

Second, a market-value tax base is more neutral than a book-value base,
in the specific sense that it eliminates the asset-specific variation in
effective tax rates that the book-value channel creates. However, this does
not mean that market-value taxation eliminates all non-neutrality: the
liquidity and dividend channels operate regardless of the tax base
specification.

Third, the Norwegian practice of valuation discounts for unlisted assets
and real estate has the effect of partially compensating for the real costs
imposed by illiquidity and dividend extraction. While this compensation is
not formally calibrated to the liquidity costs, it operates in the right
direction: assets with the highest illiquidity costs (private firms, real
estate) receive the largest valuation discounts ($\theta \ll 1$).

Fourth, the interaction between the channels suggests that a well-designed
wealth tax would explicitly account for illiquidity costs in setting
valuation discounts, rather than relying on historical book values that may
or may not approximate the relevant cost. This would require estimating
asset-specific illiquidity costs~$c_j$---a substantial empirical
undertaking, but one that would improve the alignment between the effective
tax burden and the actual cost of compliance.

% =========================================================================
\section{Conclusion and Further Work}\label{sec:conclusion}
% =========================================================================

This paper establishes that a proportional wealth tax levied at a uniform
rate on the market value of all assets is neutral with respect to four
dimensions of the investor's problem: the risk-reward profile of wealth
(CV invariance), the optimal portfolio weights (including the tangency
portfolio), the geometric relationship between the tax and portfolio choice
(orthogonality), and the per-share asset price (pricing neutrality). These
results are derived first under geometric Brownian motion and then
generalised to any return distribution in the location-scale family; the
two distributional assumptions that matter---finite second moments and
proportionality of the tax---are minimal.

The neutrality rests on two conditions: universal taxation at market
value, and frictionless markets. The payment mechanism---whether the tax
is funded from dividends, share sales, or other taxable wealth---does
not affect the result, provided that reinvested dividends remain within
the tax base (Appendix~\ref{app:dividends}). A complementary analysis under the CAPM
(\Cref{sec:capm_specialcase}) confirms that after-tax betas equal pre-tax
betas, the security market line contracts uniformly, and---under CRRA
preferences---general equilibrium returns and prices are identical to the
no-tax benchmark. This analysis also identifies a pricing error in
\citet{Fama2021}, who adds the wealth tax to the cost of capital without
recognising that the entire opportunity set, including the risk-free rate,
is contracted by the same factor.

When these conditions fail, three channels of non-neutrality emerge
with opposing price effects: book-value taxation raises valuations for
assets with $\theta < 1$; liquidity frictions and dividend extraction lower
them. The net effect is asset-class specific and depends on institutional
details---book-to-market ratios, illiquidity costs, and the profitability
of foregone investment.

\medskip
\noindent\textbf{Directions for further work.}
Several extensions would strengthen and complement the analysis.

On the \emph{theoretical} side, the model could be embedded in a
multi-period consumption-saving framework with endogenous labour supply,
where the wealth tax interacts with the intertemporal allocation of
resources. A general equilibrium analysis with heterogeneous agents facing
different effective tax rates---as arises under progressive taxation or
asset-class exemptions---would connect the Security Market Fan mechanism of
\Cref{sec:capm_smf} to equilibrium asset pricing with realistic tax
codes. The interaction between the wealth tax and income or capital gains
taxes is unexplored in our framework and relevant for optimal tax design.

On the \emph{empirical} side, four avenues stand out. First, calibrating
asset-class-specific illiquidity costs~$c_j$ (listed equities, private
firms, real estate) would allow a quantitative assessment of the liquidity
channel. Second, Norwegian registry data, which link individual portfolio
holdings to tax records, could be used to test whether portfolio weights
are indeed invariant to the wealth tax rate---exploiting cross-sectional
variation in effective rates created by valuation discounts. Third, the
book-value channel can be estimated using variation in $\theta$ across
Norwegian firms combined with the pricing formulas of
\Cref{prop:bv_multiperiod}. Fourth, the dividend-extraction channel
predicts that firms whose owners face binding payout constraints exhibit
higher payout ratios around tax dates---a prediction testable with payout
data.

\medskip
\noindent\textbf{A note on structure.}
To streamline the main argument, three self-contained discussions have been
placed in appendices: the detailed classification of distributional
assumptions (\Cref{app:distributional}), the continuous-time versus
discrete-time comparison (\Cref{app:ct_dt}), and the survey of related
literature (\Cref{app:literature}). These can be read independently without
disrupting the main narrative.

% =========================================================================
% Acknowledgements
% =========================================================================

\subsection*{Acknowledgements}
The author acknowledges the use of Claude (Anthropic) for assistance with
literature review, \LaTeX{} typesetting, mathematical exposition, and
editorial refinement, and Lemma (Axiomatic AI) for review and proof
checking. All substantive arguments, economic reasoning, and conclusions
are the author's own.

% =========================================================================
% References
% =========================================================================

\bibliographystyle{plainnat}
%\bibliography{../references/master}

% =========================================================================
\appendix

\section{Alternative Tax Payment Mechanisms: Dividends vs.\ Share
Sales}\label{app:dividends}

The main results (Sections~\ref{sec:single_gbm}--\ref{sec:generalisation})
assume that the wealth tax is paid by selling a fraction $\tw$ of the
investor's position at each period end---the proportional-dilution mechanism.
This appendix considers an alternative: the wealth tax is paid from dividend
income rather than from share sales. We retain all other assumptions of the main text,
including the absence of income and capital gains taxes.

\subsection{Setup}\label{app:div_setup}

At each period end $i$, the asset pays a dividend $D_i$ per share. The dividend
yield is $\delta_i \equiv D_i / P_i$. The wealth tax liability is
$\tw N_{i-1} P_i$, and the investor uses dividend income $D_i N_{i-1}$ to pay
as much of this liability as possible before resorting to share sales.

Three cases arise, depending on the relationship between the dividend yield and
the tax rate.

\subsection{Case 1: Dividends Always Sufficient
  ($\delta_i \geq \tw$ with Certainty)}\label{app:div_case1}

If the dividend yield exceeds the tax rate in every period, the tax is fully
paid from dividends and no shares are sold:
\begin{equation}
  N_n = N_0 \quad \text{(constant)},
  \qquad
  W_n^A = N_0 P_n = W_n^B.
\end{equation}
The equity wealth of the taxed investor is \textbf{identical} to that of the
untaxed investor. All four propositions hold trivially---not because of
multiplicative separability, but because the tax does not touch the capital
base at all.

The economic content is different from the proportional-dilution case. Under
proportional dilution, the tax erodes the number of shares but the investor
consumes the full dividend $D_i$ per share. Under dividend payment, the
investor retains all shares but consumes only $(D_i - \tw P_i)$ per share. The
tax has effectively become a \textbf{consumption tax on dividend income}: it
reduces the investor's standard of living without affecting the trajectory of
equity wealth.

This distinction matters for welfare analysis (the investor's lifetime utility
from consumption is lower) but not for the portfolio choice, risk-reward, or
pricing results, which concern equity wealth and asset returns.

\subsection{Case 2: Dividends Insufficient, Constant Yield
  ($\delta < \tw$)}\label{app:div_case2}

If the dividend yield is constant at $\delta$ (dropping the time subscript) and
below the tax rate, all dividends go to tax and the shortfall
$(\tw - \delta) P_i$ per share is covered by selling
shares. The share count evolves as:
\begin{equation}
  N_i = N_{i-1} \bigl(1 - (\tw - \delta)\bigr),
  \qquad
  N_n = N_0 \bigl(1 - (\tw - \delta)\bigr)^n
\end{equation}
and equity wealth is:
\begin{equation}
  W_n^A = N_0 \bigl(1 - (\tw - \delta)\bigr)^n P_n
        = \bigl(1 - (\tw - \delta)\bigr)^n \cdot W_n^B.
\end{equation}
This is exactly the proportional-dilution model of
Sections~\ref{sec:single_gbm}--\ref{sec:generalisation} with the
\textbf{effective tax rate}
$\tw^{\mathrm{eff}} = \tw - \delta$ replacing $\tw$. The multiplicative
separability is preserved, and all four propositions hold with
$\tw^{\mathrm{eff}}$ in place of~$\tw$. The dividend acts as a partial shield
against capital erosion: the higher the yield, the smaller the effective
dilution rate.

In the limiting case $\delta = \tw$, the effective rate is zero and we recover
Case~1. In the limiting case $\delta = 0$ (no dividends), we recover the
original proportional-dilution model.

\subsection{Case 3: Stochastic Dividend Yield}\label{app:div_case3}

If the dividend yield $\delta_i = D_i / P_i$ varies stochastically across
periods, the effective liquidation rate at each period is
$\max(\tw - \delta_i, \, 0)$, and the share count evolves as:
\begin{equation}\label{eq:Nn_stochastic}
  N_n = N_0 \prod_{i=1}^{n} \bigl(1 - \max(\tw - \delta_i, \, 0)\bigr).
\end{equation}
The product $\prod_i (1 - \max(\tw - \delta_i, 0))$ is now a
\textbf{stochastic} quantity---it depends on the realised path of dividend
yields, which are themselves functions of the price process. The deterministic
multiplicative scalar $(1 - \tw)^n$ of the proportional-dilution model is replaced
by a random variable that is correlated with the cumulative return
$G^{(n)} = P_n / P_0$.

This breaks the multiplicative separability $W_n^A = c \cdot W_n^B$ with
deterministic~$c$, because the tax factor and the return are no longer
independent. Specifically:
\begin{equation}
  W_n^A = N_0 P_n \prod_{i=1}^{n}
    \bigl(1 - \max(\tw - \delta_i, \, 0)\bigr).
\end{equation}
The product
$P_n \cdot \prod_i (1 - \max(\tw - \delta_i, 0))$ is a product of correlated
random variables, and in general:
\begin{align}
  E[W_n^A] &\neq \Bigl(\prod_{i=1}^n
    E\!\bigl[1 - \max(\tw - \delta_i, 0)\bigr]\Bigr) \cdot E[W_n^B], \\
  \SD(W_n^A) &\neq \Bigl(\prod_{i=1}^n
    E\!\bigl[1 - \max(\tw - \delta_i, 0)\bigr]\Bigr) \cdot \SD(W_n^B).
\end{align}
The coefficient of variation of $W_n^A$ is therefore \textbf{not} in general
equal to that of $W_n^B$, and \Cref{prop:cv_gen} (CV invariance) may fail.
Similarly, the portfolio choice results
(\Cref{prop:portfolio_gen,prop:orthogonality_gen}) may be affected, because the
stochastic tax factor introduces an additional source of portfolio-level risk
that depends on the interaction between dividend yields and asset returns.

\subsection{Pricing Neutrality Under All Cases}\label{app:div_pricing}

\Cref{prop:pricing_gen} (pricing neutrality) is robust to the payment
mechanism. The value of a share to the investor depends on the discounted
stream of after-tax cash flows. Regardless of whether the tax is paid from
dividends or from share sales, the total tax liability is
$\tw \times \text{market value}$---the method of payment is a financing
decision, not a valuation one. In partial equilibrium, the taxed investor's
discount rate adjusts to reflect the tax, as shown by
\citet{BjerksundSchjelderup2022}, and the NPV of the share remains equal to
the market price for both investors.

More concretely: if the investor pays the tax from dividends, the after-tax
dividend per share is $(D_i - \tw P_i)$ instead of $D_i$, but the discount
rate also falls from $\rf$ to $\rf^A = \rf - \tw(1+\rf)$. These
effects cancel in the discounted cash flow calculation, preserving pricing
neutrality. Note, however, that this analysis concerns the \emph{financial
mechanics} of tax payment. When the tax forces dividend extraction that
displaces profitable investment, a \emph{real} cost arises that breaks
pricing neutrality; see \Cref{sec:div_extraction}.

\subsection{Summary}\label{app:div_summary}

\begin{table}[ht]
\centering
\caption{Effect of tax payment mechanism on results.}
\label{tab:payment}
\small
\begin{tabular}{@{}llccc@{}}
\toprule
Payment mechanism & $N_n$ & Mult.\ sep. & Props $1'$--$3'$ & Prop $4'$ \\
\midrule
Proportional dilution (share sales)
  & $N_0(1-\tw)^n$, det. & Yes & Hold & Holds \\
Dividends suff.\ ($\delta \geq \tw$)
  & $N_0$, const. & Trivially & Trivially & Holds \\
Div.\ insuff., const.\ $\delta < \tw$
  & $N_0(1-\tw^{\mathrm{eff}})^n$, det.
  & Yes & Hold (adj.) & Holds \\
Stochastic div.\ yield
  & Path-dep., stoch. & \textbf{Breaks} & \textbf{May fail} & Holds \\
\bottomrule
\end{tabular}
\end{table}

The proportional-dilution assumption of the main text is not merely a modelling
convenience. It ensures that the tax factor $(1 - \tw)^n$ is deterministic and
independent of the return realisation, which is the structural property that
drives Propositions~$1'$--$3'$. Alternative payment mechanisms preserve the
results if and only if the effective dilution rate remains deterministic.
Among the cases considered, this holds when the dividend yield is either zero
(original model), constant, or always above the tax rate. When the dividend
yield is stochastic and sometimes insufficient, the interaction between the tax
payment and the return process introduces a new source of risk that breaks the
multiplicative structure.

% =========================================================================

\section{The Role of Distributional Assumptions}\label{app:distributional}

The location-scale family serves a specific and limited role in the
generalisation. It is worth being precise about what it contributes and where
it is not needed.

\textbf{What the location-scale family provides.} The assumption contributes
two things. First, \emph{two-moment sufficiency}: \citet{Meyer1987} establishes
that if all alternatives belong to the same location-scale family, then
expected utility can be expressed as
$E[U(X)] = V(\mu_X, \sigma_X)$, justifying the mean-variance framework
without assuming quadratic utility; see also \citet[Ch.~4]{Ingersoll1987}.
Second, \emph{closure under linear combination}: in the multivariate case
(elliptical distributions), any portfolio return inherits the distributional
family of the individual asset returns
\citep[Theorem~2]{HamadaValdez2008}, so that
$R_P = \w^\top \R$ has the same shape parameter as~$\R$.

These two services are needed for Propositions~$2'$ and~$3'$, which require
the $(\sigma, \mu)$ representation of the efficient frontier and the geometric
contraction argument. They are \emph{not} needed for CV invariance
(Proposition~$1'$), which follows from scalar multiplication alone; for
tangency portfolio invariance (Proposition~$2'$a), which is an algebraic
identity; or for pricing neutrality (Proposition~$4'$), which is entirely
distribution-free.

\textbf{The elliptical hierarchy.} The relevant distributional classes form a
hierarchy:

\begin{table}[ht]
\centering
\caption{Distributional classes and their properties.}
\label{tab:elliptical}
\small
\begin{tabular}{@{}llcc@{}}
\toprule
Class & Univariate examples & Multivariate closure & Portfolio separation \\
\midrule
Normal / GBM       & Normal, lognormal         & Yes & Yes \\
Elliptical         & Student-$t$, logistic, Laplace & Yes & Yes \\
Location-scale     & All symmetric unimodal    & Requires elliptical & Requires elliptical \\
Finite 2nd moments & Any with $E[R^2] < \infty$ & Not necessarily & SR invariance only \\
\bottomrule
\end{tabular}

\smallskip
\noindent\footnotesize
Portfolio separation for elliptical distributions follows
\citet{OwenRabinovitch1983}.
\end{table}

Our results sit most naturally in the elliptical class, which is the broadest
family for which the full portfolio theory apparatus (mean-variance optimality,
two-fund separation) is known to hold.

\textbf{Summary.} \Cref{tab:assumptions} classifies the assumptions by their
role in each result.

\begin{table}[ht]
\centering
\caption{Assumptions required for each result.}
\label{tab:assumptions}
\small
\begin{tabular}{@{}lcccc@{}}
\toprule
Result & A1 (moments) & A2 (loc-scale) & A3 (prop.\ tax)
  & A4 (part.\ equil.) \\
\midrule
Prop $1'$: CV invariance
  & Required & Not required & \textbf{Essential} & Not required \\
Prop $2'$a: Tangency inv.
  & Required & Not required & \textbf{Essential} & Not required \\
Prop $2'$b: Full weight inv.
  & Required & MV only & \textbf{Essential} & Not required \\
Prop $3'$: Orthogonality
  & Required & Required & \textbf{Essential} & Not required \\
Prop $4'$: Pricing neut.
  & Not req. & Not required & \textbf{Essential} & \textbf{Essential} \\
\bottomrule
\end{tabular}

\smallskip
\noindent\footnotesize
``MV only'' means A2 is needed for mean-variance preferences but not for CRRA.
\end{table}

The single assumption that is essential to all four results is
\textbf{proportionality}---the tax is a fixed fraction of market value, applied
uniformly to all assets. This is what generates the multiplicative separability
$W_n^A = (1 - \tw)^n W_n^B$, which is the structural foundation for
everything that follows. The distributional assumption, by contrast, is a
regularity condition that ensures the relevant moments and geometric objects
are well-defined. It plays no substantive economic role.

% =========================================================================

\section{Continuous Time vs.\ Discrete Time}\label{app:ct_dt}

The continuous-time GBM formulation
(Sections~\ref{sec:single_gbm}--\ref{sec:pricing_gbm}) and the discrete-time
generalisation (\Cref{sec:generalisation}) yield the same qualitative results
but differ in detail:

\begin{table}[ht]
\centering
\caption{Continuous-time vs.\ discrete-time comparison.}
\label{tab:ct_dt}
\small
\begin{tabular}{@{}lll@{}}
\toprule
Feature & Continuous time (GBM) & Discrete time (general) \\
\midrule
Tax on drift & Additive: $\mu_W = \mu_P - \tw$
  & Multiplicative: $\mu_W = (1-\tw)(1+\mu_P) - 1$ \\
Tax on volatility & None: $\sigma_W = \sigma_P$
  & Scaling: $\sigma_W = (1-\tw)\sigma_P$ \\
Sharpe ratio & Invariant & Invariant \\
Tangency portfolio & Invariant & Invariant \\
Full weight vector & Invariant for all $f(\mu_W,\sigma_W)$
  & CRRA: invariant; MV: scales by $\frac{1}{1-\tw}$ \\
Geometry in $(\sigma,\mu)$ & Vertical translation & Homothetic contraction \\
\bottomrule
\end{tabular}
\end{table}

The continuous-time results are the natural limit of the discrete-time results
as the period length shrinks.

% =========================================================================

\section{Relation to the Existing Literature}\label{app:literature}

\textbf{\citet{BjerksundSchjelderup2022}} derive \Cref{prop:pricing_gen}
(pricing neutrality) in a discrete-time DCF framework without specifying a
return distribution. Their result that
$\mathrm{NPV}_0^D = \mathrm{NPV}_0^F = c_0/r$ is the pricing counterpart of
our framework, and rests on no-arbitrage rather than on any specific asset
pricing model (see \Cref{sec:independence_capm}). Our contribution is to embed
this in a unified treatment that also covers the portfolio choice and
risk-reward dimensions. A technical difference concerns the after-tax discount
rate.  \citeauthor{BjerksundSchjelderup2022} assume a beginning-of-period tax
base, for which the after-tax cost of capital is $k - \tw$
\citep{Kruschwitz2023}.  Our model assesses the tax at period end on current
market values, yielding $k^A = (1-\tw)(1+k)-1 = k - \tw - \tw k$
(\Cref{eq:mm_kA}).  The two expressions differ by the cross term~$\tw k$,
which is small for realistic parameters ($\tw = 1\%$, $k = 8\%$ gives
$\tw k = 0.08\%$) but is needed for no-arbitrage under end-of-period
assessment.  In both cases, the pricing neutrality result $V = V^0$
obtains because the tax reduces the numerator and denominator of the
valuation in the same proportion; the timing convention determines only the
form of the discount rate, not the conclusion.

\textbf{\citet{EiksethLindset2009}} develop a CAPM-like framework with
heterogeneous asset taxes across investors. They derive an after-tax beta equal
to the pre-tax beta multiplied by an asset-specific tax adjustment, and show
that the Security Market Line becomes a Security Market Fan. Their setting is
more general than ours (different tax rates on different assets), but in the
special case of a uniform proportional tax, their results are consistent with
the portfolio invariance and orthogonality we derive here.
\Cref{sec:capm_smf} shows that book-value taxation
(\Cref{sec:bv_taxation}) creates exactly the heterogeneous effective rates
that produce the Security Market Fan, providing the equilibrium mechanism
behind the beta-dependent GE pricing formula in
\Cref{prop:bv_multiperiod}(b).

\textbf{\citet{Fama2021}} argues that wealth taxes lower asset prices by
raising the pre-tax required return: in a Gordon perpetuity,
$P = D/(r + \tw) < D/r$. His analysis adds the wealth tax to the cost of
capital but does not adjust the discount rate for the fact that the entire
opportunity set---including the risk-free asset---is also taxed.
\Cref{sec:fama_error} shows that when the wealth tax applies universally, the
correct after-tax discount rate is $k^A = (1-\tw)(1+k) - 1 < k$, and the
price effect vanishes ($V = V^0$). Fama's result holds only when the marginal
investor has access to an untaxed outside option, making it a
partial-equilibrium statement despite its general-equilibrium framing.

\textbf{\citet{JohnsenLensberg2014}} model the wealth tax as a separate claim
on the firm using a Modigliani-Miller framework. For listed firms taxed on
market value, they derive $V = V^0(1 - \tw/k)$, where~$k$ is the pre-tax
cost of capital. As discussed in \Cref{sec:mm_perspective}, this formula
reduces the cash flows by the tax but retains the pre-tax discount rate;
using the after-tax discount rate recovers $V = V^0$, consistent with our
pricing neutrality result. Their analysis of non-listed firms, where the tax
base is book value rather than market value, identifies a leverage-like
increase in equity beta that is specific to the Norwegian institutional
setting and lies outside our framework of proportional taxation on market
value.

\textbf{\citet{Stowe2021}} analyses the capitalisation of wealth taxes into
stock and bond valuations, showing that taxed assets will have lower valuations
and higher required rates of return. His focus on valuation impacts across
asset classes complements our partial equilibrium analysis of pricing
neutrality between taxed and untaxed investors.

\textbf{\citet{HamadaValdez2008}} prove that the CAPM holds when returns
follow multivariate elliptical distributions, using a generalised Stein's
lemma. Their result shows that the elliptical class---within which our
Propositions~$2'$ and~$3'$ sit most naturally---is compatible with the full
portfolio theory apparatus. As discussed in \Cref{sec:independence_capm}, our
results do not \emph{require} CAPM; the Hamada--Valdez result establishes that
working within the elliptical class is consistent with it.

\textbf{\citet{Sandmo1985}} and \textbf{\citet{Stiglitz1969}} analyse the
effects of taxation on risk-taking in a broader setting that includes income
taxes, capital gains taxes, and their interaction with portfolio choice. The
orthogonality result (\Cref{prop:orthogonality_gen}) is specific to the
proportional wealth tax---income taxes operate on returns, not on the quantity
of assets, and therefore interact with portfolio choice in fundamentally
different ways.

The GBM framework and optimal portfolio choice follow
\textbf{\citet{Merton1969,Merton1971}}. The mean-variance portfolio theory
follows \textbf{\citet{Markowitz1952}}. The interaction of taxation with
portfolio choice has been studied by, among others,
\textbf{\citet{Constantinides1983}} and
\textbf{\citet{DammonSpattZhang2001}}, though the specific results on
proportional wealth taxation and orthogonality presented here do not appear to
be widely documented in the existing literature.

On the empirical side, \textbf{\citet{BjornebyEtAl2023}} find a positive
causal relationship between wealth tax liability and employment in Norwegian
closely held firms.  Their proposed mechanism---that intangible assets in
non-traded firms are effectively tax-exempt, incentivising owners to invest in
human capital within their businesses rather than in taxed financial
assets---is consistent with the non-uniform assessment channel discussed in
\Cref{sec:bv_taxation} and provides empirical support for the neutrality
benchmark: when effective tax rates differ across assets, investors reallocate
toward the less-taxed margin rather than reducing economic activity.

% =========================================================================
\end{document}